\documentclass[12pt, onecolumn, draftcls, english]{IEEEtran}%
\usepackage[fleqn]{amsmath}
\usepackage{graphicx}
\usepackage{amssymb}
\usepackage{amsfonts}
\usepackage{amscd}
\usepackage{multirow}
\usepackage{color}
\usepackage{babel}
\usepackage{subfigure}
\usepackage{listings}
\usepackage{epstopdf}
\usepackage{epsfig}
\usepackage{algpseudocode}
\usepackage{array}
\usepackage{psfrag}
\usepackage{color}
\newtheorem{theorem}{Theorem}
\newtheorem{lemma}[theorem]{Lemma}
\newtheorem{remark}[theorem]{Remark}

\newtheorem{rem}{Remark}[section]

\begin{document}
\title{Source-Channel Coding under Energy, Delay and Buffer Constraints }
\author{\IEEEauthorblockN{Oner Orhan\textsuperscript{1}, Deniz G{\"u}nd{\"u}z\textsuperscript{2}, and Elza Erkip\textsuperscript{1}}

\IEEEauthorblockA{\textsuperscript{1}Dept. of ECE, NYU Polytechnic School of Engineering Brooklyn, NY, USA\\
\textsuperscript{2}Imperial College London, London, UK}
}
\maketitle
\renewcommand{\thefootnote}{\fnsymbol{footnote}}
\footnotetext{This work was presented in part at the IEEE International Symposium on Information Theory, Istanbul, Turkey, Jul. 2013.}

\begin{abstract}
Source-channel coding for an energy limited wireless sensor node is investigated. The sensor node observes independent Gaussian source samples with variances changing over time slots and transmits to a destination over a flat fading channel. The fading is constant during each time slot. The compressed samples are stored in a finite size data buffer and need to be delivered in at most $d$ time slots. The objective is to design optimal transmission policies, namely, optimal power and distortion allocation, over the time slots such that the average distortion at destination is minimized. In particular, optimal transmission policies with various energy constraints are studied. First, a battery operated system in which sensor node has a finite amount of energy at the beginning of transmission is investigated. Then, the impact of energy harvesting, energy cost of processing and sampling are considered. For each energy constraint, a convex optimization problem is formulated, and the properties of optimal transmission policies are identified. For the strict delay case, $d=1$, $2D$ waterfilling interpretation is provided. Numerical results are presented to illustrate the structure of the optimal transmission policy, to analyze the effect of delay constraints, data buffer size, energy harvesting, processing and sampling costs.
\end{abstract}

\section{Introduction}\label{intro}

Wireless sensor nodes measure physical phenomena, compress their measurements and transmit the compressed data to a destination such that the reconstruction distortion at the destination is minimized subject to delay constraints. Various components of a wireless sensor node consume energy, including sensing, processing and communications modules. The small size and low cost of typical sensors impose restrictions on the available energy, size of the battery and data buffers, and efficiency of sensing and transmission circuity. When the variation of the physical environment and the communication channel are also considered, the optimum management of available energy is essential to ensure minimal reconstruction distortion at the destination under limited resources.

We consider a wireless sensor node that collects samples of a Gaussian source and delivers them to a destination. To model the time-varying nature of the source and the channel, we consider a time slotted system such that the source variance and the channel power gain remain constant within each time slot that spans $n$ uses of the channel. We assume that the source samples arrive at the beginning of each time slot and need to be delivered within $d$ time slots. The data buffer, which stores the compressed samples, has finite capacity. We first assume that the sensor node is run by a battery and energy is only consumed for data transmission. Our goal is to identify the optimal power and compression rate/distortion allocation over a finite time horizon such that the average distortion at the destination is minimized. This problem is formulated under the offline optimization framework, that is, we assume that the sensor node knows all the source variances and channel gains of time slots a priori. We show that this problem can be cast into the convex optimization framework which allows us to identify the necessary and sufficient conditions for the optimal power and distortion allocation. For the special case of strict delay constraints, i.e., $d=1$, we show that the optimal strategy has a {\em two-dimensional (2D) waterfilling} interpretation.

We then extend the above model to study various energy constraints on the sensor node. First, we investigate energy harvesting, and consider a model in which energy arrives (or becomes available) at the beginning of each time slot. Then, we concentrate on various sources of energy consumption in the sensor such as the operation of transmitter circuitry (digital-to-analog converters, mixers, filters) and the sensing components (source acquisition, sampling, quantization, and compression). We model the former energy cost by the processing cost $\epsilon_p$  Joules per channel use, and the latter by the sampling cost $\epsilon_s$ Joules per sample. We consider that these energy costs are constant and independent of the transmission power. The offline optimization framework retains its convexity under energy harvesting, processing and sampling costs. Accordingly, we identify properties of the optimal power and distortion allocation when the processing and sampling costs are considered.

In recent years optimal energy management polices for source-channel coding has received significant attention. Optimal energy allocation to minimize total distortion using uncoded analog transmission is investigated in \cite{Poor}, \cite{limmane}. In \cite{Poor}, the total distortion is minimized under power constraint by using a best linear unbiased estimator at the fusion center. In \cite{limmane}, distortion minimization for energy harvesting wireless nodes under finite and infinite energy storage is studied for both causal and non-causal side information about channel gains and energy arrivals. For separate source and channel coding in an energy harvesting transmitter, optimal energy allocation is investigated in \cite{osvaldo}-\cite{xi}. In \cite{osvaldo}, compression and transmission rates are jointly optimized for stochastic energy arrivals taking into consideration the energy used for source compression. The work in \cite{Matz} extends results in \cite{osvaldo} to incorporate battery and memory constraints. Our previous work \cite{oner} considers delay limited transmission of a time varying Gaussian source over a fading channel with infinite memory size. The problem of sensing and transmission for parallel Gaussian sources for a battery operated transmitter with processing and sensing costs is studied in \cite{xi}. In \cite{Seyedi}, maximization of the number of samples delivered with only the sampling cost is studied.

There is also a rich literature on energy harvesting transmission policies for throughput optimization ignoring the source coding aspects, such as \cite{deniz3}-\cite{deniz2}, \cite{elza}-\cite{elza2}. In \cite{deniz3}, overview of recent developments in the energy harvesting transmission policies is provided. In \cite{Yang2012}, Yang and Ulukus investigate offline throughput maximization and transmission completion time minimization problems over a constant channel. The throughput maximization problems for single fading link \cite{fade}-\cite{Ho}, broadcast \cite{broad} and multiple access channels \cite{multi} have also been studied. In \cite{deniz2}, an energy harvesting system is studied under  battery constraints, such as battery leakage and limited size. In short range communications, as in wireless sensor networks, sensing and processing cost can be comparable to transmission cost \cite{Cui}, \cite{Asanovic}. Recently, the effect of processing cost on the throughput maximizing policies are studied for parallel Gaussian channels in \cite{glue}, and in the energy harvesting scenario, for a single-link in \cite{elza}-\cite{Nossek}, and for a broadband channel in \cite{elza2}.

The paper is organized as follows. In the next section, we describe the system model. In Section \ref{ss:single_energy}, we investigate distortion minimization for a battery-run system, and provide properties of the optimal distortion and power allocation. We also propose a $2D$ waterfilling algorithm for $d=1$. We study distortion minimization with energy constraints in Section \ref{ss:energy const}. We investigate the structure of the optimal distortion and power allocation, and provide $2D$ directional waterfilling algorithm for the energy harvesting, processing and sampling cost in Sections \ref{ss:multi_energy}, \ref{process}, \ref{sampling} respectively. In Section \ref{result}, numerical results are presented and in Section \ref{s:conc} we conclude.

\vspace{-0.1in}
\section{System Model}\label{sys model}
We consider a wireless sensor node measuring source samples that are independent and identically distributed (i.i.d.) with a given distribution. Due to the potentially time-varying nature of the underlying physical phenomena, we assume that the statistical properties of the source samples change over time. To model this change, we consider a time slotted system with $N$ time slots, with each time slot containing $n$ source samples. We denote the samples arriving at time slot $i$ as source $i$, and assume that the samples of source $i$ come from a zero-mean Gaussian distribution with variance $\sigma_i^2$. The samples are compressed and stored in a data buffer of size $B_{max}$ bits/source sample. In addition, in order to model delay-limited scenarios, e.g., real-time applications, we impose delay constraints on the samples, such that samples arriving in a time slot need to be delivered within at most $d$ time slots. After $d$ time slots, samples become stale, and we set the corresponding distortion to its maximum value, $\sigma_i^2$.

We consider that the collected samples are delivered over a fading channel having an additive white Gaussian noise (AWGN) with zero mean and unit variance. We assume that the real valued channel power gain remains constant within each time slot, and its value for time slot $i$ is denoted by $h_i$. Assuming that the time slot durations in terms of channel use are large enough to invoke Shannon capacity arguments, the maximum transmission rate in time slot $i$ is given by the Shannon capacity $\frac{1}{2}\log(1+h_ip_i)$, where $p_i$ indicates the average transmission power in time slot $i$. Since the source statistics do not change within a time slot, constant power transmission within each time slot can be shown to be optimal. This follows from the concavity and the monotonically increasing property of the Shannon capacity. We also assume that in each time slot the number of source samples collected is equal to the number of channel uses. However, the results in this paper can be easily extended to bandwidth expansion/compression.

Since the samples are continuous valued, lossy reconstruction at the destination is unavoidable. We consider mean squared error distortion criterion on the samples at the destination. Denoting the average distortion of the source $i$ by $D_i$, the objective is to minimize $D \triangleq \sum_{i=1}^N D_i$. We are interested in \textit{offline optimization}, that is, we assume that the transmitter knows all the sample variances and the channel gains for time slots $i=1,...,N$ in advance.
A \emph{transmission policy} refers to average transmission power $p_i$ and average distortion $D_i$ allocation to channel $i$ and source samples collected in time slot $i$, respectively, for $i=1,...,N$. We study the optimal transmission policy under different energy constraints. First, we consider a battery operated system in which sensor node has $E$ Joules of energy at the beginning of transmission. Then, we investigate more stringent energy constraints including energy harvesting, energy cost of processing and sampling. For the energy harvesting system, we assume that the sensor harvests energy packets of size $E_i$ Joules at the beginning of time slot $i$, $i=1,...,N$. The processing cost is modelled as constant $\epsilon_p$ Joules per transmitted symbol, and it is assumed to be independent of the transmission power. The sampling cost is also assumed to be constant, and considered as $\epsilon_s$ Joules per source sample and independent of the sampling rate \cite{osvaldo}.

\begin{figure}\label{sysmod}
\centering
\includegraphics[scale=0.5,trim= 25 45 0 10]{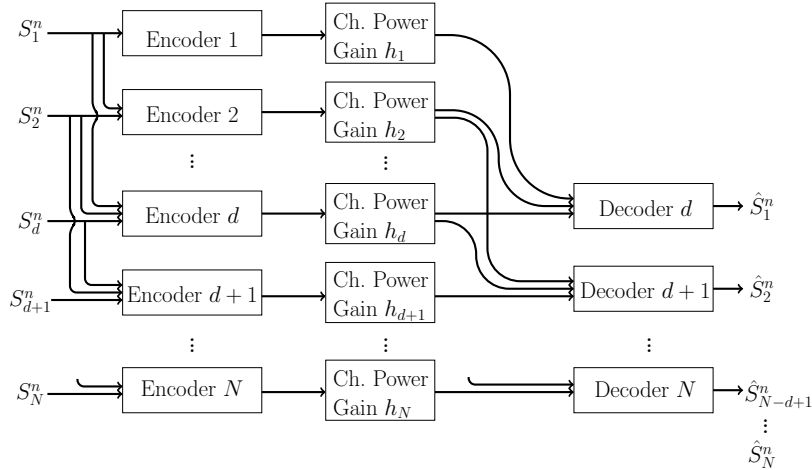}
\caption{Equivalent multiterminal source-channel communication scenario under orthogonal multiple access. $S_i^n$ denotes source samples in time slot $i$, $\hat{S}_i^n$ denotes their reconstruction at the receiver.}
\vspace{-0.4in}
\end{figure}

This formulation considers separate source and channel coding. We can equivalently model this point-to-point communication problem as multiterminal source-channel communication under orthogonal multiple access as shown in Figure 1. In this correspondence, Encoder $i$ corresponds to the encoder at time slot $i$ which observes source samples over the last $d$ time slots, and transmits over the channel within time slot $i$. Similarly, we can consider a separate decoder for each time slot $i$, $i=d,d+1,...,N$, such that Decoder $i$ observes channel outputs $i-d+1,...,d$, and reconstructs the source samples that have been accumulated within time slot $i-(d-1)$. Note that this is equivalent to decoding the source samples just before their deadline expires, since decoding them earlier does not gain anything to the system. Using \cite{Shamai} we can argue the optimality of source-channel separation in this setting; hence the above formulation gives us the optimal total distortion.

In the next section, we study the optimal distortion and power allocation for the battery-run system. Then, in Section \ref{ss:energy const} we investigate additional energy constraints on the system including energy harvesting and energy cost of processing and sampling. We study each energy constraint separately to study its effect on the optimal transmission policy. In Section \ref{ss:multi_energy}, we incorporate energy harvesting capability into the sensor node. Then, in Section \ref{process}, we consider jointly the energy cost of transmission and processing. Finally, we consider both transmission and sampling energy cost in Section \ref{sampling}. Details of the energy models will be presented in the relevant sections.
\vspace{-0.1in}
\section{Distortion Minimization for a Battery-Run System}\label{ss:single_energy}
\vspace{-0.1in}
We assume that the sensor node has $E$ Joules of energy at the beginning of transmission. We focus only on the energy consumption of the power amplifier, and ignore any energy cost due to processing and sampling.  We denote the rate allocated to source $i$ in time slot $j$, $j\leq N$ as $R_{i,j}$. Note that $R_{i,j}=0$ for $i+d<j$ or $j<i$. In a feasible transmission policy, the transmission power in time slot $i$ limits the maximum rate that can be transmitted over that time slot. Therefore, any feasible transmission policy should satisfy the following constraints:
\vspace{-0.1in}
\begin{align}\label{rate 1}
\sum_{i=j-d+1}^{j}{R_{i,j}} &\leq \frac{1}{2}\log\left(1+h_j p_j\right), \quad j=1,...,N,
\vspace{-0.1in}
\end{align}
where $R_{i,j}=0$ for $i<1$. The rate-distortion theorem in \cite{infotheory} states that the average distortion of the samples taken at time slot $i$, $D_i$, should satisfy the following.
\begin{align}\label{rate 2}
\frac{1}{2}\log\left(\frac{\sigma_i^2}{D_i}\right) &\leq \sum_{j=i}^{i+d-1}{R_{i,j}}, \quad i=1,...,N.
\vspace{-0.1in}
\end{align}
In addition, the limited data buffer size imposes the following constraints.
\begin{align}\label{rate 3}
\sum_{j=k}^{k+d-1}{\sum_{i=j-d+1}^{k}R_{i,j}} & \leq B_{max}, \quad k=1,...,N.
\vspace{-0.1in}
\end{align}
\vspace{-0.1in}
\begin{remark}
Note that the buffer size constraint is in terms of the total bits per sample for those sources that have not yet expired. This would mean that the buffer size is infinite since the above assumptions of capacity and rate-distortion achieving codes stipulate $n\rightarrow \infty$.
\end{remark}
The goal is to identify $R_{i,j}$ and $D_i$ values that minimize $D = \sum_{i=1}^N D_i$ under constraints (\ref{rate 1})-(\ref{rate 3}).

It can be shown using Fourier-Motzkin elimination \cite{gamal} that the above inequalities (\ref{rate 1})-(\ref{rate 3}) are equivalent to the following causality, delay and rate constraints, respectively. The proof of Fourier-Motzkin elimination for the case of three time slots with delay constraint $d=2$ is given in Appendix.
\vspace{-0.05in}
\begin{align}
\label{const 3}
\sum_{j=i}^{N}{r_j} &\leq  \sum_{j=i}^{N}{c_j}, \quad i=1,...,N,\\
\label{const 4}
\sum_{j=k}^{i}{r_j} &\leq \sum_{j=k}^{i+d-1}{c_j}, \quad i=k,...,N-d,\quad k=1,...,N-d,\\
\label{const 5}
\sum_{j=k}^{i+1}{r_j} & \leq  \sum_{j=k}^{i}{c_j}+ B_{max}, ~~ i=k,...,N-1, ~~k=1,...,N-1, \\ \label{const 6}
r_j &\leq  B_{max}, \quad i=1,...,N,
\end{align}
where $r_i \triangleq \frac{1}{2}\log\left(\frac{\sigma_i^2}{D_i}\right)$ and $c_i \triangleq \frac{1}{2}\log\left(1+h_i p_i\right)$. Notice that $r_i$ corresponds to the total source rate for the samples collected in time slot $i$, and $c_i$ is the channel capacity for time slot $i$ for power $p_i$ and channel gain $h_i$.
The causality constraints in (\ref{const 3}) suggest that the samples can only be transmitted after they have arrived. The delay constraints in (\ref{const 4}) stipulate that the samples collected in time slot $i$ need to be delivered to the destination within the following $d$ time slots. The data buffer constraints in (\ref{const 5})-(\ref{const 6}) impose restrictions on the amount of bits per sample. The goal of the transmitter is to allocate its transmission power $p_i$ within each time slot and choose distortion level $D_i$ for each source, $i=1,...,N$, such that the causality, delay, and data buffer constraints are satisfied, while the sum distortion $D$ at the destination is minimized.

Then, the optimization problem can be formulated as follows.
\vspace{-0.05in}
\begin{subequations}\label{p:2}
\begin{align}\label{p:2a}
&\underset{r_i,c_i}{\operatorname{min}} && \sum_{i=1}^{N}{\sigma_i^2 2^{-2r_i}} \\\label{p:2b}
&\text{s.t.} && \sum_{i=1}^{N}  \frac{2^{2c_i}-1}{h_i} \leq E,  \\ \label{p:2c}
&&& \sum_{j=i}^{N}{r_j} \leq \sum_{j=i}^{N}{c_j}, \quad i=1,...,N,  \\ \label{p:2d}
&&& \sum_{j=k}^{i}{r_j} \leq \sum_{j=k}^{i+d-1}{c_j},  \quad i=k,...,N-d,  \quad k=1,...,N-d,\\ \label{p:2d2}
&&& \sum_{j=k}^{i+1}{r_j} \leq \sum_{j=k}^{i}{c_j}+  B_{max}, \quad i=k,...,N-1,  \quad k=1,...,N-1, \\ \label{p:2e}
&&& 0 \leq r_i \leq B_{max} \quad \text{and } \quad 0\leq c_i, ~~ i=1,...,N.
\vspace{-0.1in}
\end{align}
\end{subequations}
where the constraint in (\ref{p:2b}) ensures that the total consumed energy is less than the energy available in the battery at $t=0$. The constraints in (\ref{p:2c}), (\ref{p:2d}), and (\ref{p:2d2}) are the causality, delay and data buffer size constraints from (\ref{const 3}), (\ref{const 4}), and (\ref{const 5}), respectively. Since the optimization problem in (\ref{p:2}) is convex, we can compute the optimal solution by efficient numerical methods \cite{Boyd}. In the following, we investigate the properties of the optimal solution using the Karush-Kuhn-Tucker (KKT) optimality conditions. The Lagrangian of (\ref{p:2}) is defined as follows:
\begin{align}\label{l:1}
\mathcal{L} & = \sum_{i=1}^{N}{\sigma_i^2 2^{-2r_i}}+\lambda \left( \sum_{i=1}^{N}  \frac{2^{2c_i}-1}{h_i}- E_1 \right)+\sum_{i=1}^{N}{\gamma_i \left( \sum_{j=i}^{N}{r_j}-\sum_{j=i}^{N}{c_j}\right)} \nonumber\\
&+\sum_{k=1}^{N-d}\sum_{i=k}^{N-d}{\delta_{i,k} \left(\sum_{j=k}^{i}{r_j}- \sum_{j=k}^{i+d-1}{c_j}\right)}+\sum_{k=1}^{N-1}\sum_{i=k}^{N-1}{\zeta_{i,k} \left(\sum_{j=k}^{i+1}{r_j}- \sum_{j=k}^{i}{c_j}-B_{max}\right)}\nonumber\\
& -\sum_{i=1}^N \beta_i r_i+ \sum_{i=1}^N\rho_i (r_i-B_{max})-\sum_{i=1}^{N} \mu_i c_i,
\end{align}
where $\lambda \geq 0$, $\gamma_i\geq 0$, $\delta_{i,k}\geq 0$, $\zeta_{i,k}\geq 0$, $\beta_i \geq 0$, $\rho_i \geq 0$ and $\mu_i \geq 0$ are Lagrange multipliers corresponding to (\ref{p:2b})-(\ref{p:2e}).

Taking the derivative of the Lagrangian with respect to $r_{i}$ and $c_i$, we get
\begin{eqnarray}\label{d:1}
\frac{\partial \mathcal{L}}{\partial r_{i}} = - 2 (\ln2) \sigma_i^2 2^{-2r_i}+ \sum_{j=1}^{i}{\gamma_j}+\sum_{k=1}^{i}\sum_{j=i}^{N-d}{\delta_{j,k}}+\sum_{k=1}^{i}\sum_{j=i-1}^{N-1}{\zeta_{j,k}}-\beta_i+\rho_i=0, \quad \forall i,
\end{eqnarray}
where $\zeta_{i-1,i}=0$ for $\forall i$, and
\begin{eqnarray}\label{d:2}
\frac{\partial \mathcal{L}}{\partial c_{i}} = \lambda \frac{2 (\ln2) 2^{2c_i}}{h_i}  - \sum_{j=1}^{i}{\gamma_j} -\sum_{k=1}^{i}\sum_{j=i-d+1}^{N-d}{\delta_{j,k}}-\sum_{k=1}^{i}\sum_{j=i}^{N-1}{\zeta_{j,k}} -\mu_i= 0,\quad  \forall i,
\end{eqnarray}
where  $\delta_{j,k}=0$ for $j< k$.

\subsection{Optimal Distortion Allocation}
From (\ref{d:1}), replacing $r_i$ with $\frac{1}{2}\log\left(\frac{\sigma_i^2}{D_i^*}\right)$, we obtain
\begin{eqnarray}\label{ds:1}
D_i^*=\frac{1}{2\ln 2}\left(\sum_{j=1}^{i}{\gamma_j}+\sum_{k=1}^{i}\sum_{j=i}^{N-d}{\delta_{j,k}}+\sum_{k=1}^{i}\sum_{j=i-1}^{N-1}{\zeta_{j,k}}-\beta_i+\rho_i\right).
\end{eqnarray}
The complementary slackness conditions require that, whenever $\beta_i >0$, we have $D_i=\sigma_i^2$, and whenever $\rho_i >0$, we have $D_i=\sigma_i^2 2^{-2B_{max}}$. Therefore, the optimal distortion $D_i$ can be further simplified as
\begin{eqnarray}\label{sol 1}
D_i^*=\left\{
\begin{array}{c l}
    \sigma_i^2 2^{-2B_{max}}, &  ~~~\text{if } \xi_i \leq \sigma_i^2 2^{-2B_{max}}, \\
    \xi_i, &  ~~~\text{if }\sigma_i^2 2^{-2B_{max}} < \xi_i< \sigma_i^2, \\
   \sigma_i^2,  &  ~~~\text{if } \xi_i \geq \sigma_i^2,
\end{array}\right.
\end{eqnarray}
where $\xi_i$ is defined as:
\begin{eqnarray}\label{sol 2}
\xi_i \triangleq \frac{1}{2 \ln2}\left(\sum_{j=1}^{i}{\gamma_j}+\sum_{k=1}^{i}\sum_{j=i}^{N-d}{\delta_{j,k}}+\sum_{k=1}^{i}\sum_{j=i-1}^{N-1}{\zeta_{j,k}}\right).
\end{eqnarray}
Note that $\xi_i$ is similar to the {\em reverse water level} in the classical solution of the optimal distortion levels for parallel Gaussian sources \cite{infotheory}. While the classical solution has a fixed reverse water level, i.e., $\xi_i$ is independent of $i$, in our formulation, due to the causality, delay and data buffer size constraints, the reverse water level depends on the source index $i$. Note that the optimal distortion $D_i$ is confined to the interval  $[\sigma_i^2 2^{-2B_{max}}, \sigma_i^2]$ for time slot $i$.

Next, we identify some properties of the optimal distortion allocation.
\begin{lemma}\label{lemma 1}
Whenever the reverse water level $\xi_i$ in (\ref{sol 2}) increases from time slot $i$ to time slot $i+1$, all samples collected until time slot $i$ must be transmitted by the end of time slot $i$, and whenever $\xi_i$ decreases from time slot $i$ to time slot $i+1$, either the data buffer is full at the beginning of time slot $i$ and/or delivery of the samples collected at time slot  $k$, $k \in i+1,...,i+d-2$, is postponed by $i-k+d$ time slots.
\end{lemma}
\begin{proof}
From (\ref{sol 2}), we have
\begin{eqnarray}\label{sol 3}
\xi_{i+1}-\xi_i=\frac{\gamma_{i+1}+\sum_{j=i+1}^{N-d}\delta_{j,i+1}+\sum_{j=i+1}^{N-1}\zeta_{j,i+1}-\sum_{k=1}^{i-1}\zeta_{i-1,k}-\sum_{k=1}^{i}\delta_{i,k}}{2 \ln2}, ~i=1,...,N-1.
\end{eqnarray}
Therefore, when $\xi_{i+1}-\xi_i>0$, either $\gamma_{i+1}$ or, for some $j\geq i$ $\delta_{j,i+1}$ or $\zeta_{j,i+1}$ , must be positive. From the complementary slackness conditions, we know that whenever $\gamma_{i+1}>0$, the constraint in (\ref{p:2c}) is satisfied with equality, i.e., $\sum_{j=i+1}^{N}r_j=\sum_{j=i+1}^{N}c_j$. This means that all samples collected until time slot $i$ must be transmitted by the end of time slot $i$ since the later time slots can only support the source rates $r_j$, $j \geq i+1$. In addition, from the complementary slackness conditions and the constraint in (\ref{p:2d}), we can conclude that when $\delta_{j,i+1}>0$, $\sum_{k=i+1}^{j}r_k=\sum_{k=i+1}^{j+d-1}c_k$ for $j \geq i+1$ must be satisfied. Since only samples collected at time slots $i+1,...,j$ are delivered in time slots $i+1,...,j+d-1$, and each group of source samples has a delay constraint of $d$ time slots, the samples collected until time slot $i$ should be delivered by the end of time slot $i$. Similarly, from the complementary slackness conditions and the constraint in (\ref{p:2d2}), we can argue that if $\zeta_{j,i+1}>0$ then $\sum_{k=i+1}^{j+1}r_k-\sum_{k=i+1}^{j}c_k=B_{max}$ for $ j\geq i+1$ must be satisfied. This means that the data arriving between time slots $i+1$ and $j$ leads to a full data buffer at time slot $j$ for $j \geq i+1$, so all the samples collected until time slot $i$ must be transmitted by the end of time slot $i$. Therefore, whenever $\xi_i$ in (\ref{sol 2}) increases from time slot $i$ to time slot $i+1$, all samples collected by time slot $i$ must be transmitted until the end of time slot $i$. Note that this leads to an empty data buffer at the end of time slot $i$ which follows from the positivity of $\gamma_{i+1}$, $\delta_{j,i+1}$, $\zeta_{j,i+1}$ for some $j\geq i+1$.

On the other hand, from the complementary slackness conditions and the constraint in (\ref{p:2d}), we can conclude that when $\delta_{i,k}>0$, $\sum_{j=k}^{i}r_j=\sum_{j=k}^{i+d-1}c_j$ for $k \leq i$ should be satisfied. Therefore, samples collected at time slot $i+1$ should be delayed $d$ time slots since time slots $i+1,...,i+d-1$ are allocated for the delivery of samples that have arrived at time slots $k\leq i$. Similarly, from the complementary slackness conditions and the constraint in (\ref{p:2d2}), we can argue that if $\zeta_{i-1,k}>0$ then $\sum_{j=k}^{i}r_j-\sum_{j=k}^{i-1}c_j=B_{max}$ for $ k\leq i-1$ must be satisfied. This means that the data buffer must be full at the beginning of time slot $i$. Since whenever $\xi_i$ decreases from time slot $i$ to time slot $i+1$, $\delta_{i,k}>0$ for some $k \leq i$, or $\zeta_{i-1,k}>0$ for some $ k\leq i-1$. We can conclude that whenever $\xi_i$ decreases from time slot $i$ to time slot $i+1$, either the data buffer is full at the beginning of time slot $i$ and/or the delivery of the samples collected at time slot $k$, $k \in i+1,...,i+d-2$, is postponed by $i-k+d$ time slots.
\end{proof}

\subsection{Optimal Power Allocation}
We can identify the optimal power allocation by replacing $c_i$ with $\frac{1}{2}\log\left(1+h_i p_i\right)$ in (\ref{d:2}). The optimal power allocation is given as follows.
\begin{eqnarray}\label{ds:2}
p_i^*=\left[\frac{\sum_{j=1}^{i}{\gamma_j} +\sum_{k=1}^{i}\sum_{j=i-d+1}^{N-d}{\delta_{j,k}}+\sum_{k=1}^{i}\sum_{j=i}^{N-1}{\zeta_{j,k}}}{2 (\ln2) \lambda} - \frac{1}{h_i}\right]^+,
\end{eqnarray}
where $\delta_{j,k}=0$ for $j< k$. We define $\nu_i \triangleq \frac{\sum_{j=1}^{i}{\gamma_j} +\sum_{k=1}^{i}\sum_{j=i-d+1}^{N-d}{\delta_{j,k}}+\sum_{k=1}^{i}\sum_{j=i}^{N-1}{\zeta_{j,k}}}{2 (\ln2)\lambda}$, which can be interpreted similarly to the classical waterfilling solution obtained for power allocation over parallel channels with {\em water level} being equal to $\nu_i$. Similarly to (\ref{sol 1}), $\nu_i$ depends on $i$ due to causality, delay and data buffer size constraints.

Next, we provide some properties of the optimal power allocation.
\begin{lemma}\label{lemma 3}
Whenever the water level $\nu_i$ in (\ref{sol 2}) increases from time slot $i$ to time slot $i+1$, all the samples collected until time slot $i$ must be transmitted by the end of time slot $i$, and whenever $\nu_i$ decreases from time slot $i$ to time slot $i+1$, either the data buffer is full at the beginning of time slot $i+1$ and/or the delivery of the samples collected at time slot $k$, $k \in i-d+2,...,i$, is postponed by at least $i-k+1$ time slots.
\end{lemma}
\begin{proof}
We can show that $\nu_{i+1}-\nu_i=  \frac{\gamma_{i+1}+\sum_{j=i+1}^{N-d}\delta_{j,i+1}+\sum_{j=i+1}^{N-1}\zeta_{j,i+1}-\sum_{k=1}^{i-d+1}\delta_{i-d+1,k}-\sum_{k=1}^{i}\zeta_{i,k}}{2 (\ln2) \lambda}$. Using arguments similar to the proof of Lemma \ref{lemma 1}, the proof can be completed.
\end{proof}

\begin{rem}\label{remark1}
When there is no delay constraint, i.e., $d=N$, the constraint in (\ref{p:2d}) is no longer necessary and $\delta_{i,k}=0$, $\forall i,k$. Therefore, from Lemma \ref{lemma 1} (Lemma \ref{lemma 3}), we can argue that full data buffer at the beginning of time slot $i$ ($i+1$) is the only reason of a decrease in the reverse water level $\xi_i$ (the water level $\nu_i$) from time slot $i$ to time slot $i+1$.
\end{rem}

\begin{rem}\label{remark2}
When the data buffer size is infinite, i.e., $B_{max}=\infty$, we have $\zeta_{i,k}=0$, $\forall i,k$. Following the arguments in Lemma \ref{lemma 1} (Lemma \ref{lemma 3}), we can conclude that whenever the reverse water level $\xi_i$ (the water level $\nu_i$) decreases from time slot $i$ to time slot $i+1$, delivery of the samples collected at time slot $k$, $k \in i+1,...,i+d-2$ ($k \in i-d+2,...,i$) must be postponed by $i-k+d$ time slots.
\end{rem}

\subsection{Strict delay constraint $(d=1)$}\label{energyd1}
In this section, we investigate the case in which the samples need to be transmitted within the following time slot, i.e., $d=1$. Note that this is equivalent to the problem investigated in \cite{xi} when sensing energy cost is zero. Here we provide a 2D waterfilling interpretation for the solution. The optimization problem in (\ref{p:2}) can be formulated as follows for $d=1$:
\begin{subequations}\label{ppp:4}
\begin{eqnarray}\label{ppp:4a}
\underset{c_i}{\operatorname{min}} && \sum_{i=1}^{N}{ \sigma_i^2 2^{-2c_i}} \\\label{ppp:4b}
\text{s.t.}~ && \sum_{i=1}^{N}\frac{2^{2c_i}-1}{h_i}+  \leq E,  \\ \label{ppp:4c}
&& 0 \leq c_i \leq B_{max}, ~~ i=1,...,N,
\end{eqnarray}
\end{subequations}
where $c_i=\frac{1}{2}\log\left(1+h_ip_i\right)=\frac{1}{2}\log\left(\frac{\sigma_i^2}{D_i}\right)$.

Solving the above optimization problem we find
\begin{align}\label{s:3}
p_i^*= \frac{\sigma_i}{\sqrt{h_i}} \left[\min\left\{ {\frac{2^{2B_{max}}}{\sigma_i \sqrt{h_i}}},\frac{1}{ \lambda}\right\} - \frac{1}{\sigma_i \sqrt{h_i}}\right]^+.
\end{align}
Defining $M_i \triangleq \frac{\sigma_i}{\sqrt{h_i}}$ and $K_i \triangleq \frac{1}{\sigma_i \sqrt{h_i}}$, the optimal power in (\ref{s:3}) can be written as
\begin{align}\label{s:3a}
p_i^* = M_i \left[ \min\left\{K_i 2^{2B_{max}}, \frac{1}{\lambda} \right\} - K_i\right]^+.
\end{align}
Since $\frac{1}{2}\log\left(\frac{\sigma_i^2}{D_i}\right) \leq \frac{1}{2}\log\left(1+h_i p_i\right)$ is satisfied with equality for $d=1$, from (\ref{s:3a}) the optimal distortion
$D_i^*$  is given by \vspace{-0.05in}
\begin{eqnarray}\label{s:5a}
D_i^*=\left\{
\begin{array}{c l}
    \sigma_i^2 2^{-2B_{max}}, &  ~~~\text{if } M_i\lambda \leq \sigma_i^2 2^{-2B_{max}}, \\
   M_i\lambda, &  ~~~\text{if }  \sigma_i^2 2^{-2B_{max}} < M_i\lambda < \sigma_i^2, \\
   \sigma_i^2,  &  ~~~\text{if }  M_i\lambda \geq \sigma_i^2.
\end{array}\right.
\end{eqnarray}

\begin{figure}
\centering
\subfigure[]{
\includegraphics[scale=0.54,trim= 18 0 0 0]{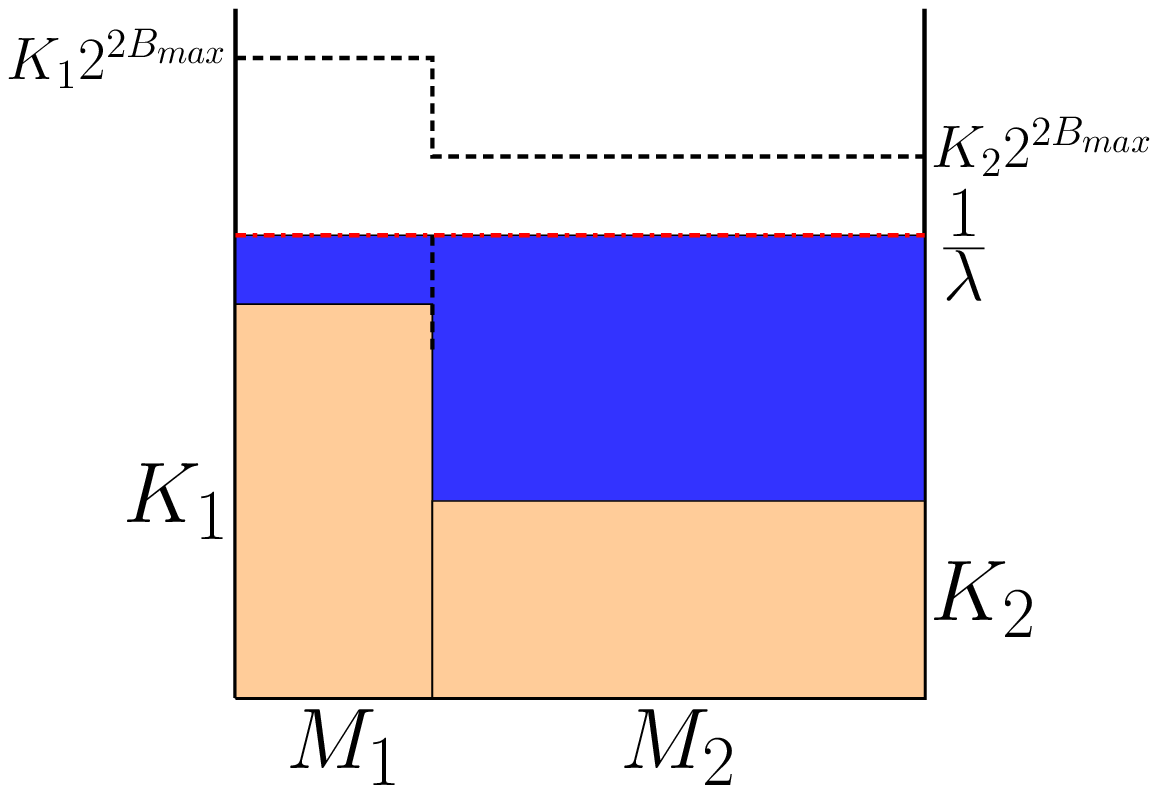}
\label{fig 1a}%
}
\subfigure[]{
\includegraphics[scale=0.54,trim= 48 0 35 0]{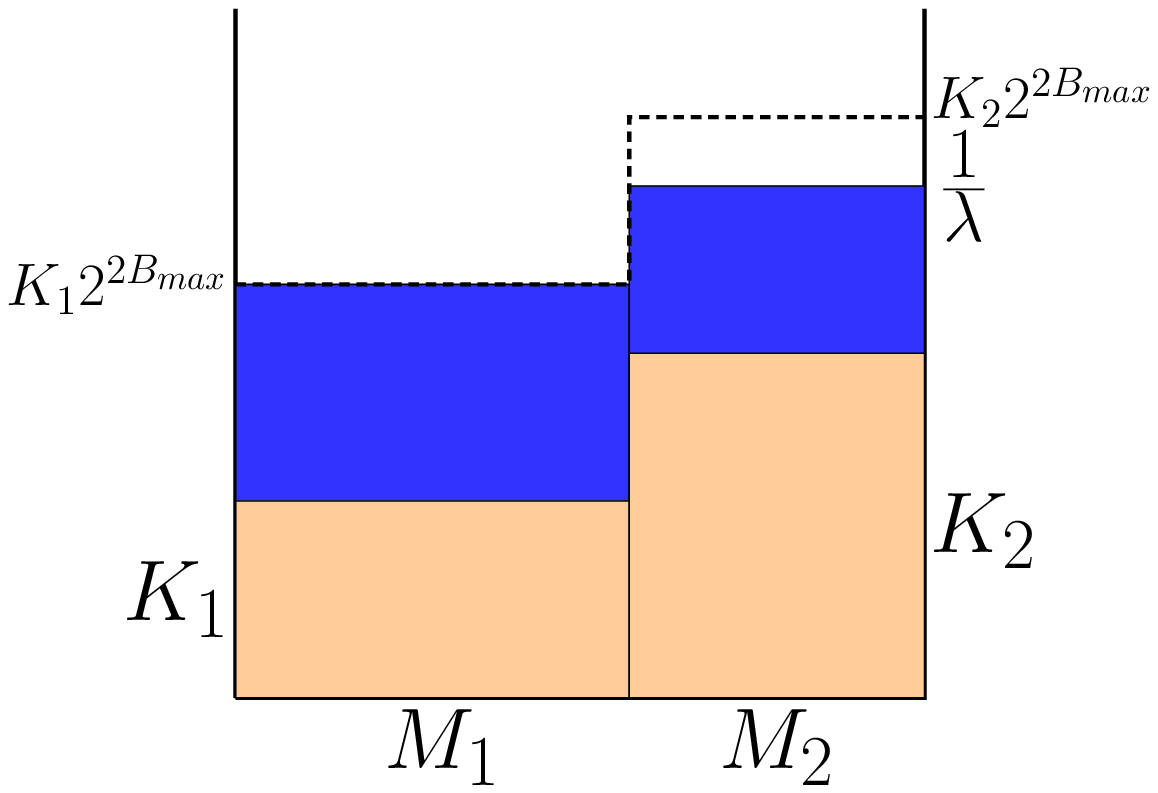}
\label{fig 1b}
\vspace{-0.1in}
}
\caption{2D water-filling algorithm, (a) data buffer constraint is not active (b) data buffer constraint is active.} \label{fig 1}
\vspace{-0.15in}
\end{figure}

The above solution is illustrated in Fig. \ref{fig 1} for $N=2$. For each time slot, we have rectangles of width $M_i$ and height $K_i$. The total energy is poured above the level $K_i$ for each time slot up to the water level  $\frac{1}{\lambda}$. The  power allocated to time slot $i$ is given by the shaded area below the water level and above $K_i$. Note that the water level is bounded by the data buffer size, i.e., $K_i 2^{2B_{max}}$, as argued in (\ref{s:3a}). If $p_i^*>0$, the distortion for source $i$ is given by the width $M_i$ times the reciprocal of the water level, and if $p_i^*=0$, the distortion for source $i$ is $\sigma_i^2=\frac{M_i}{K_i}$. As seen in Fig. \ref{fig 1a} the water level is constant over the two time slots, therefore, the optimal allocated power in time slot $i$ is given by $M_i\left(\frac{1}{\lambda}-K_i \right)$ for $i=1,2$, and the optimal distortion is given by $M_i{\lambda}$. However, in Fig \ref{fig 1b} the water level in the first time slot is limited by $K_1 2^{2B_{max}}$ due to the data buffer constraint. Therefore, as argued in Lemma \ref{lemma 3}, the increase in the water level from the first time slot to the second is due to full data buffer at the first time slot. The optimal power levels for the first and second time slots are given by $M_i K_i (2^{2B_{max}}-1)$ and $M_i\left(\frac{1}{\lambda}-K_i\right)$, respectively. The optimal average distortion values are $\frac{M_1}{K_1 2^{2B_{max}}}$ and $M_2 \lambda$ for source one and two, respectively.

\section{Distortion Minimization Under Various Energy Constraints}\label{ss:energy const}
In this section, we consider additional energy constraints on the system including energy harvesting, processing and sensing energy costs. We study the constraints separately to clearly illustrate their impact on the performance. In Section \ref{ss:multi_energy} we identify the effect of energy harvesting on the optimal power and distortion allocation. Then, in Section \ref{process} we consider the energy cost of processing circuitry together with the transmission energy, and show that the optimal power allocation is bursty in this case. Finally, in Section \ref{sampling} we investigate the effect of sampling cost on the optimal power and distortion allocation.

\subsection{Distortion Minimization with Energy Harvesting}\label{ss:multi_energy}
In this section, we consider energy harvesting at the sensor node. We consider that the sensor node harvests energy packet of size $E_i$ at the beginning of time slot $i$, $i=1,...,N$. We consider only the transmission cost and ignore the energy cost of processing and sampling. Due to energy arrivals over time, a feasible transmission policy must satisfy the following energy casuality constraint.
\begin{align}\label{p:7}
\sum_{j=1}^{i} \frac{2^{2c_j}-1}{h_j}\leq \sum_{j=1}^{i} E_j, \quad i=1,...,N.
\end{align}
Consequently, the optimization problem in (\ref{p:2}) remains the same except that the constraint (\ref{p:2b}) is replaced by the energy casuality constraints in (\ref{p:7}). Then the Lagrangian of (\ref{p:2}) with energy harvesting becomes:
\begin{align}\label{l:4}
 \mathcal{L} & = \sum_{i=1}^{N}{\sigma_i^2 2^{-2r_i}}+\sum_{i=1}^{N}{\lambda_i \left( \sum_{j=1}^{i}   \frac{2^{2c_i}-1}{h_i}-\sum_{j=1}^{i} E_j \right)}+\sum_{i=1}^{N}{\gamma_i \left( \sum_{j=i}^{N}{r_j}-\sum_{j=i}^{N}{c_j}\right)} \nonumber\\
&+\sum_{k=1}^{N-d}\sum_{i=k}^{N-d}{\delta_{i,k} \left(\sum_{j=k}^{i}{r_j}- \sum_{j=k}^{i+d-1}{c_j}\right)}+\sum_{k=1}^{N-1}\sum_{i=k}^{N-1}{\zeta_{i,k} \left(\sum_{j=k}^{i+1}{r_j}- \sum_{j=k}^{i}{c_j}-B_{max}\right)}\nonumber\\
& -\sum_{i=1}^N \beta_i r_i+ \sum_{i=1}^N\rho_i (r_i-B_{max})-\sum_{i=1}^{N} \mu_i c_i,
\end{align}
with $\lambda_i\geq 0$, $\gamma_i\geq 0$, $\delta_{i,k}\geq 0$, $\zeta_{i,k}\geq 0$, $\beta_i\geq 0$, $\rho_i\geq 0$  and $\mu_i \geq 0$ as the Lagrange multipliers.

The derivative of the Lagrangian with respect to $r_i$ is the same as in (\ref{d:1}); hence, the structure of the optimal distortion is the same as in Section \ref{ss:single_energy}. Therefore, the properties of the optimal distortion given in Lemma \ref{lemma 1} still hold.

Differentiating the Lagrangian with respect to $c_i$, we can argue that the optimal channel rate $c_i$ of time slot $i$ must satisfy
\begin{eqnarray}\label{d:33}
\frac{\partial \mathcal{L}}{\partial c_{i}} = \frac{2 (\ln2) 2^{2c_i}}{h_i}\sum_{j=i}^{N}{\lambda_j} - \sum_{j=1}^{i}{\gamma_j} -\sum_{k=1}^{i}\sum_{j=i-d+1}^{N-d}{\delta_{j,k}}-\sum_{k=1}^{i}\sum_{j=i}^{N-1}{\zeta_{j,k}} -\mu_i= 0,
\end{eqnarray}
for $i=1,...,N$ where  $\delta_{j,k}=0$ for $j< k$.

This leads to the optimal power level $p_i^*$ as follows.
\begin{eqnarray}\label{ds:3}
p_i^*=\left[\frac{\sum_{j=1}^{i}{\gamma_j}+\sum_{k=1}^{i}\sum_{j=i-d+1}^{N-d}{\delta_{j,k}}+\sum_{k=1}^{i}\sum_{j=i}^{N-1}{\zeta_{j,k}}}{2 \ln2 \sum_{j=i}^{N}{\lambda_j} } - \frac{1}{h_i}\right]^+, \quad \forall i.
\end{eqnarray}

Defining $\pi_i \triangleq \frac{\sum_{j=1}^{i}{\gamma_j}+\sum_{k=1}^{i}\sum_{j=i-d+1}^{N-d}{\delta_{j,k}}+\sum_{k=1}^{i}\sum_{j=i}^{N-1}{\zeta_{j,k}}}{2 \ln2 \sum_{j=i}^{N}{\lambda_j}}$, we can interpret (\ref{ds:3}) similarly to the directional waterfilling solution of \cite{fade} with water level equal to $\pi_i$. Accordingly, Lemma \ref{lemma 3} is updated as follows for an energy harvesting sensor node.
\begin{lemma}\label{lemma 5}
Whenever the water level $\pi_i$ in (\ref{sol 2}) increases from time slot $i$ to time slot $i+1$, either all the samples collected until time slot $i$ are transmitted by the end of time slot $i$ and/or the battery is empty at the end of time slot $i$. Similarly if $\pi_i$ decreases from time slot $i$ to time slot $i+1$, either the data buffer is full at beginning of time slot $i+1$ and/or delivery of the samples collected within time slot $k$, $k \in i-d+2,...,i$, is postponed by at least $i-k+1$ time slots.
\end{lemma}
\begin{proof}
From complementary slackness conditions, we know that when $\lambda_i>0$, the constraint in (\ref{p:7}) is satisfied with equality, hence, the battery must be empty at the end of time slot $i$. Therefore, following the arguments in the proofs of Lemma \ref{lemma 1} and \ref{lemma 3}, the proof can be completed.
\end{proof}

For the case of strict delay constraint, $d=1$, we can reformulate the optimization problem in (\ref{ppp:4}) by replacing the constraint (\ref{ppp:4b}) by (\ref{p:7}). Solving the optimization problem, we obtain the optimal transmission power and distortion in terms of $M_i$ and $K_i$ as follows.
\begin{eqnarray}
\label{s:4a}
p_i^* = M_i \left[ \min\left\{K_i 2^{2B_{max}}, \frac{1}{\sqrt{ \sum_{i=i}^{N}{\lambda_i}}} \right\} - K_i\right]^+.
\end{eqnarray}

Similarly, the optimal distortion $D_i^*$ is given by
\begin{eqnarray}\label{s:6a}
D_i^*=\left\{
\begin{array}{c l}
    \sigma_i^2 2^{-2B_{max}}, , &  ~~~\text{if } M_i\sqrt{ \sum_{i=i}^{N}{\lambda_i}} < \sigma_i^2 2^{-2B_{max}}, \\
   M_i\sqrt{ \sum_{i=i}^{N}{\lambda_i}}, &  ~~~\text{if }  \sigma_i^2 2^{-2B_{max}} < M_i\sqrt{ \sum_{i=i}^{N}{\lambda_i}} < \sigma_i^2, \\
   \sigma_i^2,  &  ~~~\text{if }  M_i\sqrt{ \sum_{i=i}^{N}{\lambda_i}} \geq \sigma_i^2.
\end{array}\right.
\end{eqnarray}

Extending Section \ref{energyd1}, we can interpret the energy harvesting solution for $d=1$ as  \textit{directional 2D water-filling} such that the harvested energy $E_i$ can only be allocated to time slots $j>i$. Accordingly, we allocate energy to the following time slots starting from the last arriving energy and continuing backwards to the first such that the energy causality constraint is satisfied. In addition, allocated power to time slot $i$ is limited by the data buffer size and channel gain, i.e., $p_i^* \leq M_i K_i \left(2^{2B_{max}}-1\right)=\frac{1}{h_i}\left(2^{2B_{max}}-1\right)$.

\begin{figure}
\centering
\subfigure[]{
\includegraphics[scale=0.55,trim= -20 0 -20 0]{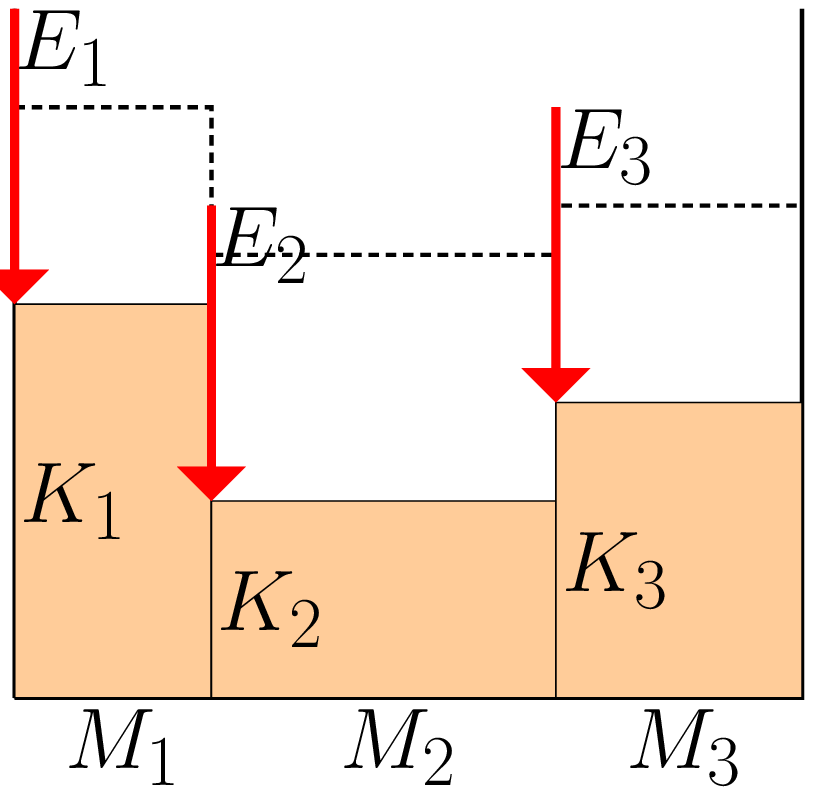}
\label{fig 2a}%
}
\subfigure[]{
\includegraphics[scale=0.55,trim= -20 0 -20 0]{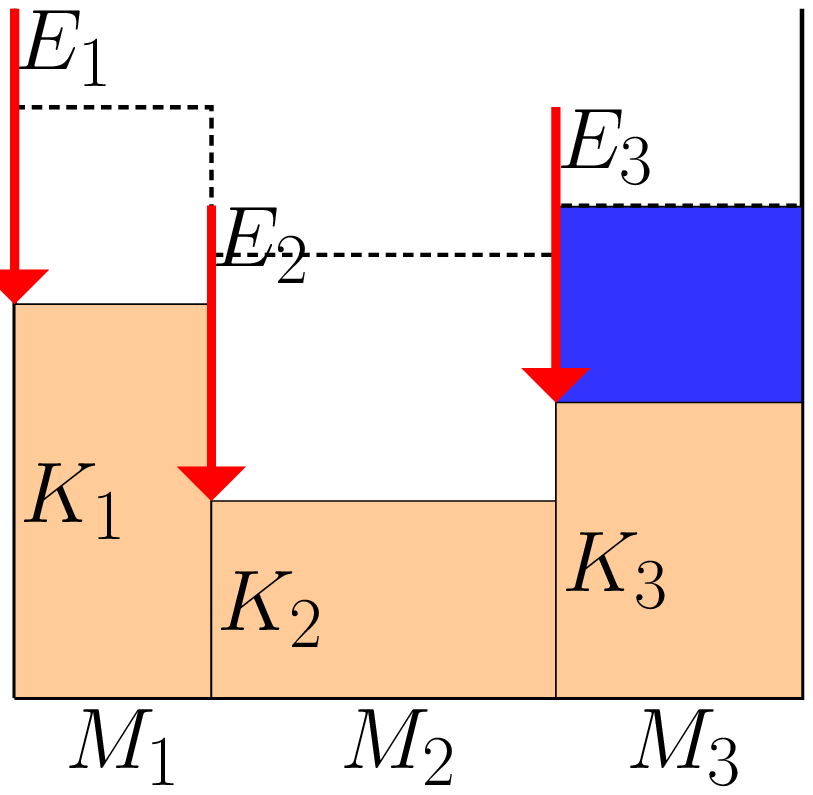}
\label{fig 2b}%
}
\subfigure[]{
\includegraphics[scale=0.55,trim= -20 0 -20 0]{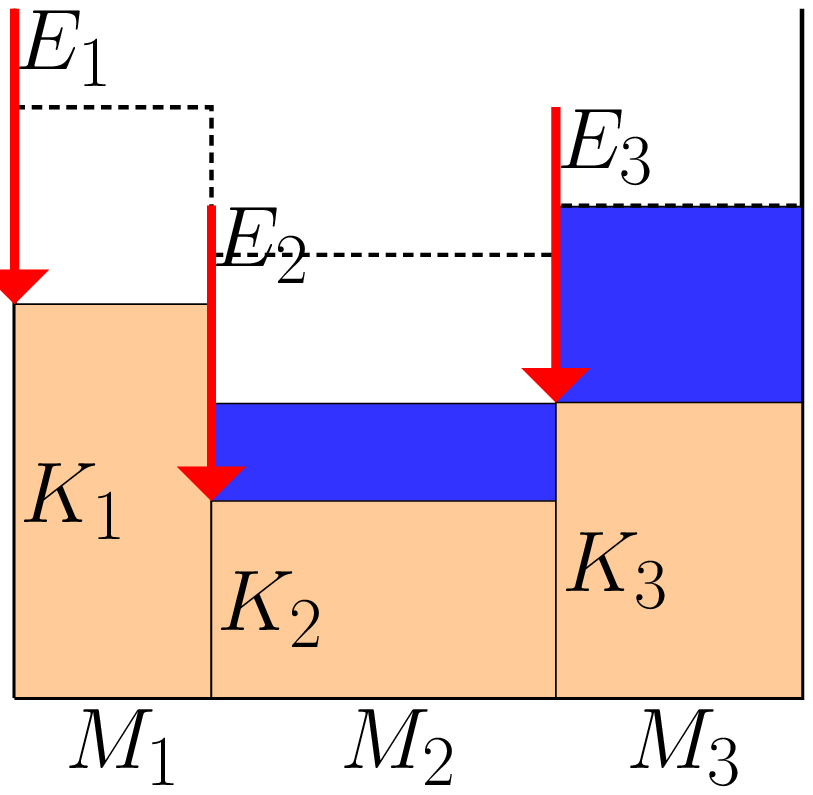}
\label{fig 2c}%
}
\subfigure[]{
\includegraphics[scale=0.55,trim= -20 0 -20 0]{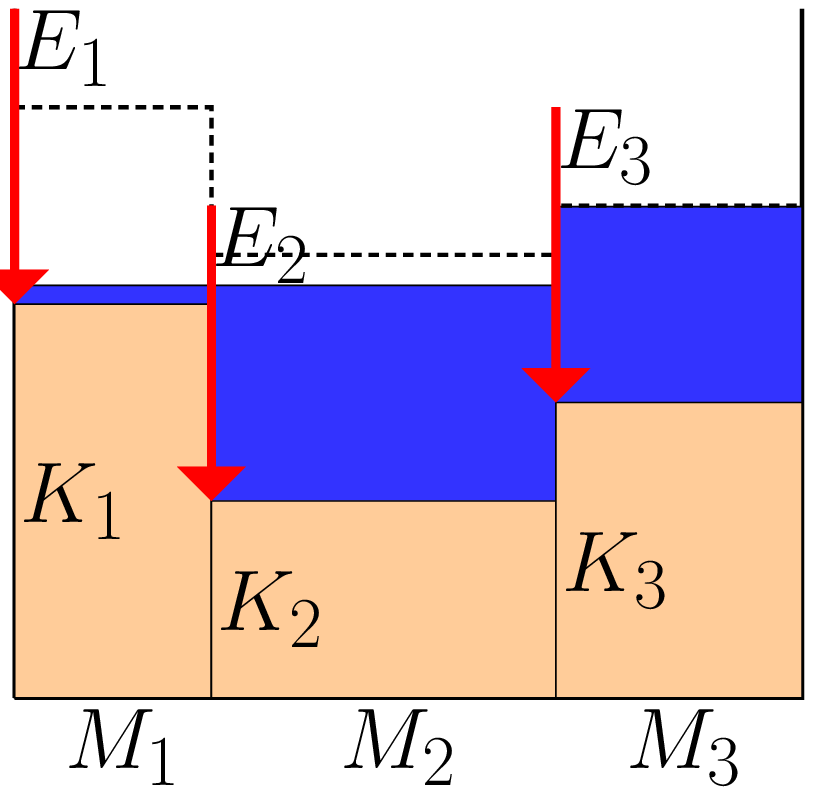}
\label{fig 2d}%
}
\caption{2D directional water-filling algorithm. Dashed line represents the buffer constraints (a) three time slots with energy arrivals $E_i$, $i=1,2,3$, (b) $E_3$ allocated to the third time slot, (c) $E_2$ allocated to the second time slot, (d) $E_1$ allocated to time slots 1 and 2.} \label{fig 2}
\vspace{-0.15in}
\end{figure}

Consider the illustration given in Fig. \ref{fig 2} with three time slots. Similarly to Fig. \ref{fig 1}, we have rectangles of width $M_i$ and height $K_i$. The horizontal dashed lines above the rectangles correspond to $K_i 2^{2B_{max}}$. The arrival times of the energy packets are represented by downward arrows. As argued above, we first allocate the last energy packet $E_3$ to the third time slot as shown in Fig. \ref{fig 2a}. Note that due to the data buffer constraint, the compression rate and the optimal power in the third time slot are limited by $B_{max}$ and $\frac{1}{h_i}\left(2^{2B_{max}}-1\right)$, respectively. This leads to an excessive energy in the battery if $E_3>\frac{1}{h_3}\left(2^{2B_{max}}-1\right)$. Then, as shown in Fig. \ref{fig 2c} the second energy packet $E_2$ is considered for time slots two and three. Since the water level of the second time slot is lower than the third time slot, $E_2$ is allocated only to the second time slot. Finally, we consider the first energy packet $E_1$ and allocate it to the first and second time slots as shown in Fig. \ref{fig 2d}. As argued before, we can obtain the optimal distortion for source $i$ by multiplying $M_i$ with the reciprocal of the water level above rectangle $i$ in Fig. \ref{fig 2d}.

\subsection{Distortion Minimization with Processing Cost}\label{process}
In this section, we investigate the properties of the optimal distortion and power allocation when, in addition to transmission energy, processing energy cost is also taken into account. For ease of exposure, we consider a battery operated system as in Section \ref{ss:single_energy} and ignore sampling cost. We assume that the sensor node consumes energy for processing only when transmitting \cite{elza}. We consider that the processing energy cost is $\epsilon_p$ Joules per transmitted symbol, and it is independent of the transmission power. As it is shown in \cite{glue}, when processing cost is taken into account, the optimal transmission policy becomes bursty. Therefore, the optimal policy may utilize only a fraction of each time slot. We denote the transmission duration within time slot $i$ by $\theta_{i}$, $0 \leq \theta_{i} \leq 1$. We redefine the auxiliary variable $c_i$, the total delivered data in time slot $i$, as $c_i \triangleq \frac{\theta_{i}}{2}\log\left(1+h_i p_i\right)$. Accordingly, the optimization problem in (\ref{p:2}) remains the same except that there is an additional constraint $0 \leq  \theta_{i} \leq 1$, and the constraint (\ref{p:2b}) is replaced by the following energy constraint.
\begin{align}\label{p:8}
\sum_{i=1}^{N} \theta_{i}\left(\frac{2^{\frac{2c_i}{\theta_{i}}}-1}{h_i}+\epsilon_p \right) \leq E.
\end{align}

Then, the Lagrangian of (\ref{p:2}) with processing energy cost is given by the following.
\begin{align}\label{lag:3}
 \mathcal{L} & = \sum_{i=1}^{N}{\sigma_i^2 2^{-2r_i}}+\lambda \left( \sum_{i=1}^{N}  \theta_{i}\left(\frac{2^{\frac{2c_i}{\theta_{i}}}-1}{h_i}+\epsilon_p \right) - E \right)+\sum_{i=1}^{N}{\gamma_i \left( \sum_{j=i}^{N}{r_j}-\sum_{j=i}^{N}{c_j}\right)} \nonumber\\
&+\sum_{k=1}^{N-d}\sum_{i=k}^{N-d}{\delta_{i,k} \left(\sum_{j=k}^{i}{r_j}- \sum_{j=k}^{i+d-1}{c_j}\right)}+\sum_{k=1}^{N-1}\sum_{i=k}^{N-1}{\zeta_{i,k} \left(\sum_{j=k}^{i+1}{r_j}- \sum_{j=k}^{i}{c_j}-B_{max}\right)}\nonumber\\
& -\sum_{i=1}^N \beta_i r_i+ \sum_{i=1}^N\rho_i (r_i-B_{max})-\sum_{i=1}^{N} \mu_i c_i -\sum_{i=1}^{N} \nu_i \theta_{i}+ \sum_{i=1}^N \phi_i (\theta_{i}-1),
\end{align}
where $\lambda \geq 0$, $\gamma_i\geq 0$, $\delta_{i,k}\geq 0$, $\zeta_{i,k}\geq 0$, $\beta_i \geq 0$, $\rho_i \geq 0$, $\mu_i \geq 0$, $\nu_i \geq 0$, and $\phi_i \geq 0$ are Lagrange multipliers.

When we take the derivative of the Lagrangian with respect to $r_i$, and replace $r_i$ with $\frac{1}{2}\log\left(\frac{\sigma_i^2}{D_i}\right)$, we obtain (\ref{ds:1}). Therefore the optimal distortion allocation satisfies (\ref{sol 1}), and the properties given in Lemma \ref{lemma 1} are also valid in this case.

Differentiating the Lagrangian with respect to $c_i$, we obtain
\begin{eqnarray}\label{dp:1}
\frac{\partial \mathcal{L}}{\partial c_{i}} = \lambda \frac{2 (\ln2) 2^{\frac{2c_i}{\theta_{i}}}}{h_i}  - \sum_{j=1}^{i}{\gamma_j} -\sum_{k=1}^{i}\sum_{j=i-d+1}^{N-d}{\delta_{j,k}}-\sum_{k=1}^{i}\sum_{j=i}^{N-1}{\zeta_{j,k}} -\mu_i= 0,~ \forall i,
\end{eqnarray}
where  $\delta_{j,k}=0$ for $j < k$. When we replace $c_i$ in the above equation with $\frac{\theta_{i}}{2}\log\left(1+h_i p_i\right)$, the optimal power allocation is given as in (\ref{ds:2}). However, unlike the optimal transmission policy in Section \ref{ss:single_energy}, due to the processing cost the optimal transmission power $p_i$ needs to be allocated $\theta_{i}$ fraction of time slot $i$. Taking derivative of the Lagrangian with respect to $\theta_{i}$, we get
\begin{eqnarray}\label{dp:2}
\frac{\partial \mathcal{L}}{\partial \theta_{i}}= \lambda \left(\frac{2^{\frac{2c_i}{\theta_{i}}}-1}{h_i}+\epsilon_p -\frac{2 (\ln2) c_i 2^{\frac{2c_i}{\theta_{i}}}}{h_i \theta_{i}}\right) - \nu_i + \psi_i = 0, \quad \forall i.
\end{eqnarray}

Using complementary slackness conditions together with (\ref{dp:2}), we can argue that
\begin{itemize}
\item If $\theta_{i}^* = 0$, then $c_i=0$ and $p_i=0$.
\item If $0< \theta_{i}^* \leq 1$, i.e., $\nu_i=0$, then assuming that $\lambda>0$, i.e., the battery is depleted by the end of time slot $N$, and replacing $c_i$ with $\frac{\theta_{i}}{2}\log\left(1+h_i p_i\right)$ in (\ref{dp:2}), we get
\begin{eqnarray}\label{e:1}
\ln2 \log(1+h_i p_i) \left(\frac{1}{h_i}+p_i\right)=(\epsilon_p + p_i)+ \frac{\psi_i}{\lambda}.
\end{eqnarray}
When $0< \theta_{i}^* < 1$, i.e., $\psi_i=0$, we obtain the same results as in \cite[Eq. (4)]{elza}. Therefore, as argued in \cite{elza}, Equation (\ref{e:1}) has a unique solution which depends only on the channel gain and the processing cost. We denote the solution of (\ref{e:1}) by $p_i^*=v_{p,i}$. When $\theta_{i}^* = 0$, i.e., $\psi_i \geq 0$, it can be argued from (\ref{e:1}) that the optimal transmission power satisfies $p_i^* \geq v_{p,i}$. Note that when $\lambda=0$, i.e., the battery may not be depleted by the end of time slot $N$, we can restrict the optimal power allocation to the above solution without loss of optimality.
\end{itemize}

Next, we study the optimal power and distortion allocation for the strict delay constraint, $d=1$. The optimization problem can be formulated by replacing the constraint (\ref{ppp:4b}) by (\ref{p:8}), and inserting an additional constraint $0\leq \theta_{i} \leq 1$. Solving the optimization problem, we obtain the optimal power allocation as follows:
\begin{eqnarray}\label{dd:16}
p_i^*=\frac{\sigma^{\frac{2}{\theta_{i}+1}}}{h_i^{\frac{\theta_{i}}{\theta_{i}+1}}}\left[\min\left\{\frac{2^{\frac{2B_{max}}{\theta_i}}}{(\sigma_i \sqrt{h_i})^{\frac{2}{1+\theta_i}}},\frac{1}{\lambda^{\frac{1}{1 +\theta_{i}}}}\right\}-\frac{1}{(\sigma_i \sqrt{h_i})^{\frac{2}{1+\theta_i}}}\right]^+,
\end{eqnarray}
where $p_i^* \geq v_{p,i}^*$. The optimal transmission duration $\theta_{i}$ satisfies the properties obtained for general delay constraint. Therefore, the optimal transmission power can be further simplified as follows:
\begin{eqnarray}\label{dd:17}
p_i^*=\left\{
\begin{array}{c l}
\frac{\sigma_i}{\sqrt{h_i}}\left[\min\left\{\frac{2^{2B_{max}}}{\sigma_i \sqrt{h_i}},\frac{1}{\sqrt{\lambda}}\right\}-\frac{1}{\sigma_i \sqrt{h_i}}\right]^+, &  ~~~\text{if }  \theta_{i}=1 , \\
v_{p,i}^*, &  ~~~\text{if } 0< \theta_{i}<1,\\
0,  &  ~~~\text{if } \theta_{i}=0.
\end{array}\right.
\end{eqnarray}

Similarly, we can argue that the optimal distortion is given as follows:
\begin{eqnarray}\label{d:5}
D_i^*=\left\{
\begin{array}{c l}
\sigma_i^2 2^{2B_{max}}, &  ~~~\text{if }  \xi_i \leq \sigma_i^2 2^{2B_{max}} \text{ and } 0<\theta_{i}, \\
\xi_i, &  ~~~\text{if } \sigma_i^2 2^{2B_{max}}<\xi_i< \sigma_i^2 \text{ and } 0< \theta_{i},\\
\sigma_i^2,  &  ~~~\text{if }  \xi_i \geq \sigma_i^2  \text{ or } \theta_{i}=0,
\end{array}\right.
\end{eqnarray}
where $\xi_i=\sigma_i^{\frac{2}{\theta_{i}+1}} \left(\frac{\lambda}{h_i}\right)^{\frac{\theta_{i}}{\theta_{i}+1}}$.

Note that for the strict delay constraint case, i.e., $d=1$, $\theta_i$ can be interpreted as the number of channel uses per source sample, or the channel-source bandwidth ratio for the source-channel pair in time slot $i$.

\subsection{Distortion Minimization with Sampling Cost}\label{sampling}
In this section, we consider sampling energy cost in addition to transmission energy. For ease of exposure, we assume a battery operated system and ignore the processing cost, i.e., $\epsilon_p=0$. Because of sampling cost, collecting all source samples may not be optimal. Hence, we assume that the sensor collects $\phi_{i}$ fraction of the samples with energy cost of $\epsilon_s$ Joules per sample. We also assume that the sampling cost is independent of the sampling rate \cite{osvaldo}. The distortion of source $i$ is now given by $D_i=\sigma_i^2(1-\phi_{i}) + \sigma_i^2 \phi_{i} 2^{-\frac{2r_i}{\phi_{i}}}$, where $r_i$ is the compression rate for the samples collected in time slot $i$. Therefore, we can obtain the corresponding optimization problem by replacing the objective function in (\ref{p:2}) with $ \sum_{i=1}^{N}{\sigma_i^2(1-\phi_{i}) + \sigma_i^2 \phi_{i} 2^{-\frac{2r_i}{\phi_{i}}}}$ and the constraint in (\ref{p:2b}) with the following energy constraint:
\begin{align}\label{p:9}
\sum_{i=1}^{N} \phi_{i}\epsilon_s + \frac{2^{2c_i}-1}{h_i} \leq E,
\end{align}
where $0 \leq \phi_{i}\leq 1$.

Accordingly, the Lagrangian of (\ref{p:2b}) with $\lambda \geq 0$, $\gamma_i\geq 0$, $\delta_{i,k}\geq 0$, $\zeta_{i,k}\geq 0$, $\beta_i \geq 0$, $\rho_i \geq 0$, $\mu_i \geq 0$, $\eta_i \geq 0$, and $\omega_i \geq 0$ as Lagrange multipliers can be written as follows:
\begin{align}\label{ll:1}
\mathcal{L} & =\sum_{i=1}^{N}{\sigma_i^2(1-\phi_{i}) + \sigma_i^2 \phi_{i} 2^{-\frac{2r_i}{\phi_{i}}}} +\lambda \left(\sum_{i=1}^{N} \phi_{i}\epsilon_s + \frac{2^{2c_i}-1}{h_i} - E \right)\nonumber \\
&+\sum_{i=1}^{N}{\gamma_i \left( \sum_{j=i}^{N}{r_j}-\sum_{j=i}^{N}{c_j}\right)} +\sum_{k=1}^{N-d}\sum_{i=k}^{N-d}{\delta_{i,k} \left(\sum_{j=k}^{i}{r_j}- \sum_{j=k}^{i+d-1}{c_j}\right)}\nonumber\\
&+\sum_{k=1}^{N-1}\sum_{i=k}^{N-1}{\zeta_{i,k} \left(\sum_{j=k}^{i+1}{r_j}- \sum_{j=k}^{i}{c_j}-B_{max}\right)}\nonumber\\
& -\sum_{i=1}^N \beta_i r_i +\sum_{i=1}^N \rho_i (r_i-B_{max}) -\sum_{i=1}^{N} \mu_i c_i -\sum_{i=1}^{N} \eta_i \phi_{i} + \sum_{i=1}^N \omega_i (\phi_{i}-1).
\end{align}

When we take the derivative of the Lagrangian with respect to $c_i$, we obtain the optimal transmission power as given in (\ref{ds:2}). Therefore, the properties provided in Lemma \ref{lemma 3} are also valid in this case. However, when we differentiate the Lagrangian with respect to $r_i$ and $\phi_{i}$, we obtain
\begin{eqnarray}\label{dd:1}
\frac{\partial \mathcal{L}}{\partial r_{i}} = -2 (\ln2) \sigma_i^2 2^{-\frac{2r_i}{\phi_{i}}} + \sum_{j=1}^{i}{\gamma_j} + \sum_{k=1}^{i}\sum_{j=i}^{N-d}{\delta_{j,k}} +\sum_{k=1}^{i}\sum_{j=i-1}^{N-1}{\zeta_{j,k}} - \beta_i +\rho_i = 0, ~ \forall i,
\end{eqnarray}
where $\zeta_{i-1,i}=0$ for $\forall i$, and
\begin{eqnarray}\label{dd:3}
\frac{\partial \mathcal{L}}{\partial \phi_{i}}= -\sigma_i^2 + \sigma_i^2 2^{-\frac{2r_i}{\phi_{i}}} +  \frac{2 (\ln2) \sigma_i^2 r_i}{\phi_{i}}2^{-\frac{2r_i}{\phi_{i}}} +\lambda \epsilon_s -\eta_i + \omega_i=0, \quad \forall i,
\end{eqnarray}
respectively.

Combining (\ref{dd:1}) with $D_i=\sigma_i^2(1-\phi_{i}) + \sigma_i^2 \phi_{i} 2^{-\frac{2r_i}{\phi_{i}}}$ we obtain the optimal distortion for source $i$ as follows:
\begin{eqnarray}\label{d:3}
D_i^*=\left\{
\begin{array}{c l}
\sigma_i^2(1-\phi_{i}) + \sigma_i^2 \phi_{i} 2^{-\frac{2B_{max}}{\phi_{i}}}, &  ~~~\text{if }  \xi_i \leq \sigma_i^2 2^{-\frac{2B_{max}}{\phi_{i}}} \text{ and } \phi_{i}>0 , \\
\sigma_i^2 \left(1-\phi_{i}\right)+ \phi_{i} \xi_i, &  ~~~\text{if }\sigma_i^2 2^{-\frac{2B_{max}}{\phi_{i}}} <\xi_i< \sigma_i^2 \text{ and } \phi_{i}>0,\\
\sigma_i^2,  &  ~~~\text{if }  \xi_i\geq \sigma_i^2  \text{ or } \phi_{i}=0,
\end{array}\right.
\end{eqnarray}
where $\xi_i$ is equal to (\ref{sol 2}). Therefore, $\xi_i$ in (\ref{d:3}) satisfies the properties given in Lemma \ref{lemma 1}. From (\ref{dd:1}) we can argue that $\xi_i= \sigma_i^2 2^{-\frac{2r_i}{\phi_{i}}}$, and from (\ref{dd:3}) we obtain:
\begin{eqnarray}\label{d:4}
\frac{\lambda \epsilon_s -\eta_i + \omega_i}{\sigma_i^2 }=1- 2^{-2k_i} -2k_i 2^{-2k_i},
\end{eqnarray}
where $k_i \triangleq \frac{r_i}{\phi_{i}}$. We can interpret $k_i$ as the compression rate for the sampled $\phi_i$ fraction of source $i$. Note that right hand side (RHS) of (\ref{d:4}) is a monotonically increasing function of $k_i$. When $0< \phi_{i}<1$, i.e., $\eta_i=0$ and $\phi_i=0$, there is a unique solution of (\ref{d:4}), which is denoted as $k_i^*=v_{s,i}$, for given $\lambda$, $\epsilon_s$, and $\sigma_i^2$. In addition, we can argue that whereas $\xi_i$ decreases as source variance $\sigma_i^2$ increases, it increases as the sampling cost increases. When $\phi_{i}=1$, i.e., $\phi_i \geq 0$, the solution of (\ref{d:4}) must satisfy $k_i^* \geq v_{s,i}$.

Next, we investigate the effect of sampling cost on the optimal power and distortion allocation in the strict delay constrained case. For $d=1$, the optimization problem can be formulated by replacing the constraint in (\ref{ppp:4b}) with (\ref{p:9}), and inserting an additional constraint $0 \leq  \phi_{i} \leq 1$. With the new objective function $\sum_{i=1}^{N}{\sigma_i^2(1-\phi_{i}) + \sigma_i^2 \phi_{i} 2^{-\frac{2c_i}{\phi_{i}}}}$, the Lagrangian of the optimization problem be can written as
\begin{align}\label{ll:2}
\mathcal{L} & =\sum_{i=1}^{N}{\sigma_i^2(1-\phi_{i}) + \sigma_i^2 \phi_{i} 2^{-\frac{2c_i}{\phi_{i}}}} +\lambda \sum_{i=1}^{N} \phi_{i}\epsilon_s + \frac{2^{\frac{2c_i}{\theta_{i}}}-1}{h_i}- E \nonumber \\
& -\sum_{i=1}^N \beta_i c_i +\sum_{i=1}^{N} \mu_i (c_i-B_{max}) -\sum_{i=1}^{N} \eta_i \phi_{i} + \sum_{i=1}^N \omega_i (\phi_{i}-1),
\end{align}
where $\lambda \geq 0$, $\beta_i\geq 0$, $\mu_i\geq 0$, $\eta_i\geq 0$, and $\omega_i\geq 0$ are Lagrange multipliers.
Differentiating the Lagrangian with respect to $c_i$ we obtain
\begin{eqnarray}\label{dd:5}
\frac{\partial \mathcal{L}}{\partial c_{i}} = - 2 (\ln2) \sigma_i^2 2^{-\frac{2c_i}{\phi_{i}}} + \frac{2 (\ln2)\lambda}{h_i} 2^{2c_i} - \beta_i + \mu_i = 0, ~ \forall i.
\end{eqnarray}
In addition, when we differentiate the Lagrangian with respect to $\phi_{i}$, we get (\ref{dd:3}).

Replacing $c_i$ in (\ref{dd:5}) with $\frac{1}{2}\log\left(1+h_i p_i\right)$, we can argue that the optimal power allocation is given by
\begin{eqnarray}\label{dd:12}
p_i^*=\frac{\sigma_i^{\frac{2\phi_{i}}{\phi_{i}+1}}}{h_i^{\frac{1}{1+\phi_{i}}}}\left[\min\left\{\frac{2^{2B_{max}}}{(\sigma_i \sqrt{h_i})^{\frac{2\theta_i}{1+\theta_i}}},\frac{1}{\lambda^{\frac{\phi_{i}}{1 +\phi_{i}}}}\right\}-\frac{1}{(\sigma_i \sqrt{h_i})^{\frac{2\theta_i}{1+\theta_i}}}\right]^+.
\end{eqnarray}

Combining (\ref{dd:5}) and (\ref{dd:3}) such that $\lambda$ is eliminated, we obtain
\begin{eqnarray}\label{dd:8}
-\sigma_i^2 + \sigma_i^2 2^{-\frac{2c_i}{\phi_{i}}} +  \frac{2 (\ln2)  \sigma_i^2 c_i}{\phi_{i}}2^{-\frac{2c_i}{\phi_{i}}} +\epsilon_s h_i \sigma_i^2 2^{-\frac{2c_i}{\phi_{i}}} 2^{-2c_i}+\beta_i - \mu_i - \eta_i + \omega_i= 0.
\end{eqnarray}
We can further simplify (\ref{dd:8}) as follows.
\begin{eqnarray}\label{dd:9}
\frac{\epsilon_s}{\frac{1}{h_i}+p_i} + \beta_i - \mu_i - \eta_i + \omega_i  =  2^{2k_i}-2(\ln2) k_i-1,
\end{eqnarray}
where $k_i=\frac{c_i}{\phi_{i}}$. Using (\ref{dd:9}), we can argue the following:
\begin{itemize}
\item If $\phi_{i}=0$ or $c_i=0$, then $p_i=0$ and $D_i=0$.
\item If $0<\phi_{i}<1$ and $0<c_i<B_{max}$, then RHS of (\ref{dd:8}) is monotonically increasing function of $k_i$, therefore Equation (\ref{dd:8}) has a unique solution $k_i^*=v_{s,i}$ for a given $\epsilon_s$, $h_i$, and $p_i$. When $h_i$ and $p_i$ are given, $c_i=\frac{1}{2}\log\left(1+h_i p_i\right)$ is known as well; and hence, we can compute the optimal sampling fraction $\phi_{i}$. Then the optimal distortion $D_i$ is given by  $D_i=\sigma_i^2(1-\phi_{i})+\sigma_i^2 \phi_{i} 2^{-2k_i}$.
\item If $\phi_{i}=1$ and $0<c_i<B_{max}$, then $\omega_i\geq 0$, therefore  from (\ref{dd:9}), we can argue that the optimal solution $k_i^*$ must satisfy $k_i^* \geq v_{s,i}$. Then, the optimal distortion  $D_i$ is given by  $D_i=\sigma_i^2(1-\phi_{i})+\sigma_i^2 \phi_{i} 2^{-2k_i}$.
\end{itemize}

\section{Illustration of the Results}\label{result}
In this section, we provide numerical results to illustrate the structure of the optimal distortion and power allocation, and to analyze the impact of the delay constraint, energy harvesting, processing and sampling costs on the optimum sum distortion. Throughout this section, we consider $N=10$ time slots. The channel gains are chosen as $\mathbf{h}=[0.4, 0.2, 0.2, 0.5, 0.4,$ $0.6, 0.9, 0.3, 0.4, 1]$, and the source variances are $\mathbf{\sigma^2}=[0.7, 0.6,$ $1, 0.5, 0.3,$ $0.6, 0.2, 0.3, 0.7, 0.5]$. We first set $d=1$ and consider a battery-run system with initial energy $E=4$ Joules. We set $\epsilon_p=\epsilon_s=0$. We illustrate the optimal rate and power allocation for $B_{max}=0.15$ bits in Fig. \ref{f:E1}. In the figure, the dashed line corresponds to $K_i 2^{2B_{max}}$. As shown in Fig. \ref{f:E1}, the data buffer size bounds the total sampled data in each time slot and the minimum distortion. The sum achievable distortion is computed as $D=4.57$. The optimal power and distortion allocation are $\mathbf{p^*}=[0.57, 0.23, 1.15, 0.46, 0.11,$ $0.38, 0.25, 0, 0.5, 0.23]$ W and $\mathbf{D^*}=[0.56, 0.57, 0.81, 0.40, 0.28, 0.48,$ $0.16, 0.3, 0.56, 0.4]$, respectively.

\begin{figure}[ht]
\center
\includegraphics[width=4in,trim= 0 0 0 0]{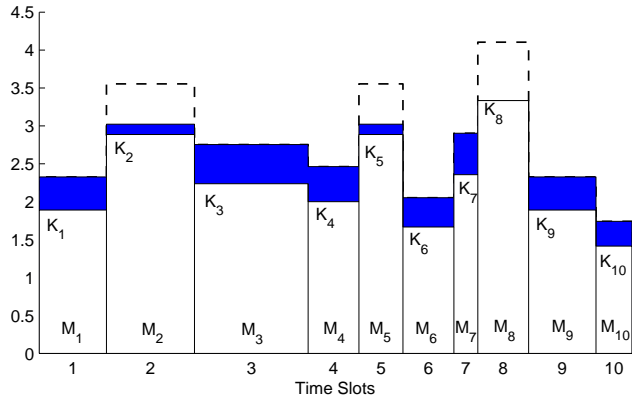}
\caption{2D waterfilling for a battery-run system. $E=4$ Joules, $B_{max}=0.15$ bits per sample, $\epsilon_p=\epsilon_s=0$, $\mathbf{h}=[0.4, 0.2, 0.2, 0.5,$ $0.4, 0.6, 0.9, 0.3, 0.4, 1]$, $\mathbf{\sigma^2}=[0.7, 0.6, 1, 0.5, 0.3, 0.6, 0.2, 0.3, 0.7, 0.5]$, $\mathbf{p^*}=[0.57, 0.23, 1.15, 0.46, 0.11, 0.38, 0.25, 0, 0.5, 0.23]$ W, and $\mathbf{D^*}=[0.56, 0.57, 0.81, 0.40, 0.28, 0.48, 0.16, 0.3, 0.56, 0.4]$.} \label{f:E1}
\end{figure}

Next, we provide the optimal rate and power allocation for the infinite data buffer size. We assume the same channel gains and source variances as given above. The 2D waterfilling solution is shown in Fig. \ref{f:D1}, resulting in the optimal total distortion $D=4.48$. The optimal power and distortion allocation are $\mathbf{p^*}=[0.74, 0, 0.48, 0.45, 0,  0.78, 0.04, 0, 0.74, 0.73]$ W and $\mathbf{D^*}=[0.53, 0.6, 0.9, 0.4, 0.3, 0.4, 0.19, 0.3, 0.53, 0.28]$, respectively.
\begin{figure}[ht]
\center
\includegraphics[width=4in,trim= 0 0 0 0]{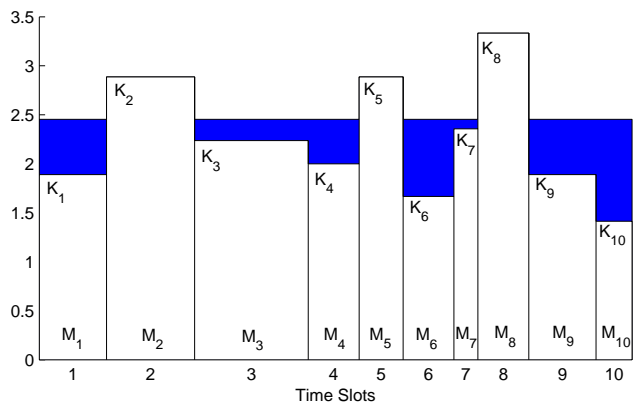}
\caption{2D waterfilling for battery-run system. $E=4$ Joules, $B_{max}\rightarrow \infty$, $\epsilon_p=\epsilon_s=0$, $\mathbf{h}=[0.4, 0.2, 0.2, 0.5,$ $0.4, 0.6, 0.9, 0.3, 0.4, 1]$, $\mathbf{\sigma^2}=[0.7, 0.6, 1, 0.5, 0.3, 0.6, 0.2, 0.3, 0.7, 0.5]$, $\mathbf{p^*}=[0.74, 0, 0.48, 0.45, 0,  0.78, 0.04, 0, 0.74, 0.73]$ W and $\mathbf{D^*}=[0.53, 0.6, 0.9, 0.4, 0.3, 0.4, 0.19, 0.3, 0.53, 0.28]$.} \label{f:D1}
\end{figure}

We illustrate the optimal distortion with respect to $B_{max}$ in Fig. \ref{f:E2}. We assume the same channel gains and source variances as before, and set $E=4$ Joules and $\epsilon_p=\epsilon_s=0$. As shown in Fig. \ref{f:E2}, the distortion decreases dramatically when the data buffer size is large. As expected, the distortion, when the delay constraint is $d=1$, is larger than the case when $d=N$. The figure also shows that the data buffer size has more impact on the distortion when the delay constraint is more relaxed. This is because a relaxed delay constraint allows more flexibility in terms of rate allocation, but this flexibility can be exploited only with a sufficiently large data buffer. In addition, distortion remains constant when the data buffer size $B_{max}\geq 0.31$ for $d=1$, and when $B_{max}\geq 1.12$ for $d=10$.
\begin{figure}[ht]
\center
\includegraphics[width=4in,trim= 0 0 0 0]{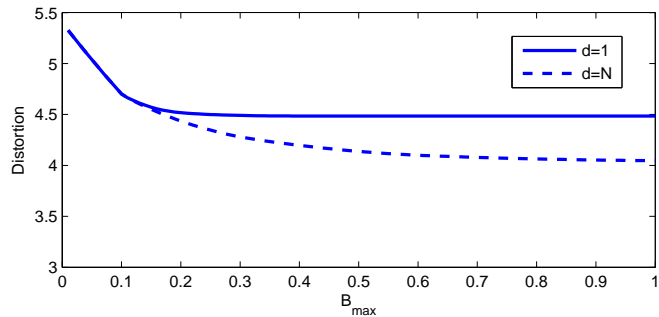}
\caption{Distortion versus buffer size. $E=4$ Joules, $\epsilon_p=\epsilon_s=0$, $\mathbf{h}=[0.4, 0.2, 0.2, 0.5,$ $0.4, 0.6, 0.9, 0.3, 0.4, 1]$, $\mathbf{\sigma^2}=[0.7, 0.6, 1, 0.5, 0.3, 0.6, 0.2, 0.3, 0.7, 0.5]$. } \label{f:E2}
\end{figure}

We investigate the variation of the optimal distortion $D$ with respect to the delay constraint $d$ in Fig. \ref{f:DD}. We consider a battery-run system with initial energy $E=4$ Joules and $\epsilon_p=\epsilon_s=0$. The optimal distortion values for increasing $d$ plotted in Fig. \ref{f:DD} show that the optimal distortion decreases monotonically for $d\leq 4$ and remains constant afterwards when $B_{max}=\infty$. However, when the data buffer size is limited to $B_{max}=0.15$ bits per sample, relaxing the delay constraint beyond two time slots does not decrease the minimum achievable distortion.
\begin{figure}[ht]
\center
\includegraphics[width=4in,trim= 0 0 0 0]{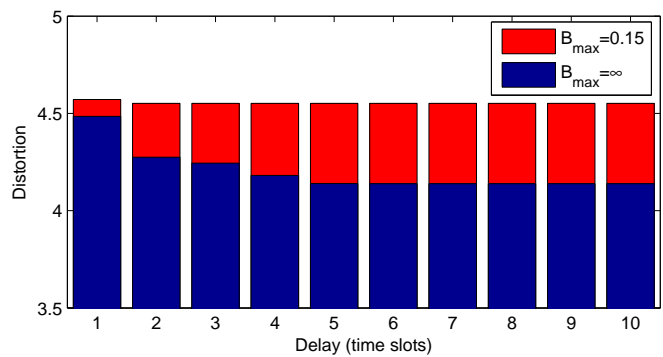}
\caption{Total distortion $D$ versus delay constraint $d$. $E=4$ Joules, $B_{max}=0.15$ bits per sample, $\epsilon_p=\epsilon_s=0$, $\mathbf{h}=[0.4, 0.2, 0.2, 0.5,$ $0.4, 0.6, 0.9, 0.3, 0.4, 1]$, $\mathbf{\sigma^2}=[0.7, 0.6, 1, 0.5, 0.3, 0.6, 0.2, 0.3, 0.7, 0.5]$. } \label{f:DD}
\end{figure}

We also investigate the variation of the optimal distortion $D$ with respect to the available energy. We consider a battery-run system with initial energy $E \in [0,10]$ Joules and $\epsilon_p=\epsilon_s=0$. We assume that $B_{max}=0.15$. As it can be seen from Fig. \ref{f:FF}, the achievable distortion decays with the available total energy, and for very low and very high energy levels, the minimum achievable distortion values are the same for $d=1$ and $d=N$. Since the allocated energy to each time slot is partly limited by the data buffer constraint,  when the available energy in the battery is large, all the samples of source $i$ can be transmitted within time slot $i$, and hence, relaxing the delay constraint does not decrease the minimum achievable distortion.
\begin{figure}[ht]
\center
\includegraphics[width=4in,trim= 0 0 0 0]{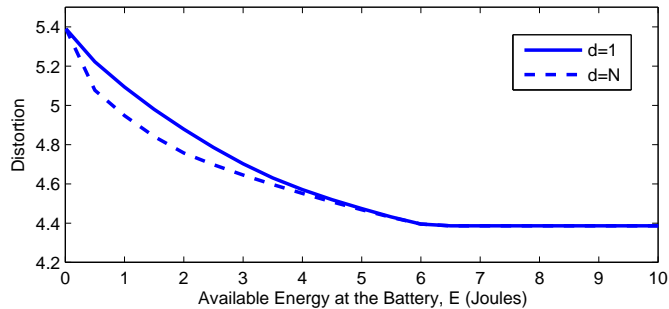}
\caption{Total distortion $D$ versus available energy. $E=4$ Joules, $\epsilon_p=\epsilon_s=0$, $\mathbf{h}=[0.4, 0.2, 0.2, 0.5,$ $0.4, 0.6, 0.9, 0.3, 0.4, 1]$, $\mathbf{\sigma^2}=[0.7, 0.6, 1, 0.5, 0.3, 0.6, 0.2, 0.3, 0.7, 0.5]$.} \label{f:FF}
\end{figure}

Next, we consider an energy harvesting system with energy packets of sizes $E_1=1, E_6=3, E_i=0$ otherwise. We set $\epsilon_p=\epsilon_s=0$ and $B_{max}=\rightarrow \infty$ bits per sample. The 2D directional waterfilling solution for infinite data buffer size is given in Fig. \ref{f:D2}. Note that the water level changes after time slot five because of directional waterfilling. The resulting optimal distortion is $D=4.50$, larger than the battery-run system with the same total energy (see Fig. \ref{f:D1}), since the battery-run system has more flexibility in allocating the available energy over time. The optimal power and distortion allocations are $\mathbf{p^*}=[0.54, 0, 0.15, 0.3, 0, 0.98, 0.13, 0, 1, 0.87]$ W and $\mathbf{D^*}=[0.57, 0.6, 0.97,   0.43, 0.3, 0.37, 0.17, 0.29, 0.49, 0.26]$, respectively.

\begin{figure}[ht]
\center
\includegraphics[width=4in,trim= 0 0 0 0]{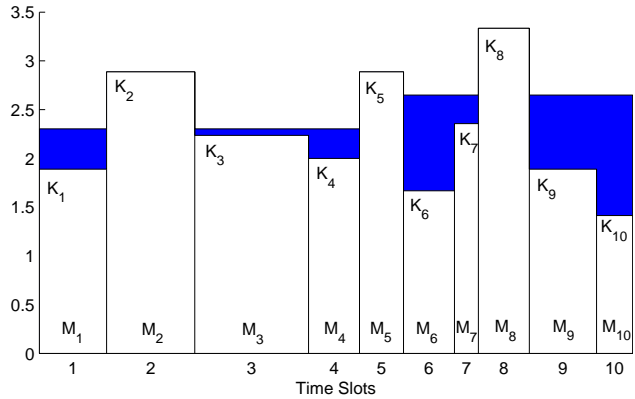}
\caption{2D directional waterfilling for an EH system. $E_1=1$, $E_6=3$, $E_i=0$ Joules, $B_{max}\rightarrow \infty$, $\epsilon_p=\epsilon_s=0$, $\mathbf{h}=[0.4, 0.2, 0.2, 0.5,$ $0.4, 0.6, 0.9, 0.3, 0.4, 1]$, $\mathbf{\sigma^2}=[0.7, 0.6, 1, 0.5, 0.3, 0.6, 0.2, 0.3, 0.7, 0.5]$, $\mathbf{p^*}=[0.54, 0, 0.15, 0.3, 0, 0.98, 0.13, 0, 1, 0.87]$ W and $\mathbf{D^*}=[0.57, 0.6, 0.97,   0.43, 0.3, 0.37, 0.17, 0.29, 0.49, 0.26]$.} \label{f:D2}
\end{figure}

The effect of the processing cost on the minimum distortion for a battery-run system is illustrated in Fig. \ref{f:HH}. We set $E=4$ Joules and $\epsilon_s=0$. As seen in the figure, when the data buffer constraint is 0.1 bits per sample and the processing cost is low, the minimum achievable distortion is the same for the delay constrained and unconstrained scenarios. However, as the processing cost increases system without delay constraint performs better than the strict delay constrained case. In addition, when the data buffer size is relaxed, the performance without a delay constraint significantly improves. However, when the processing cost is high, relaxing the data buffer size does not decrease the total distortion because high processing cost limits the compression rate.
\begin{figure}[ht]
\center
\includegraphics[width=4in,trim= 0 0 0 0]{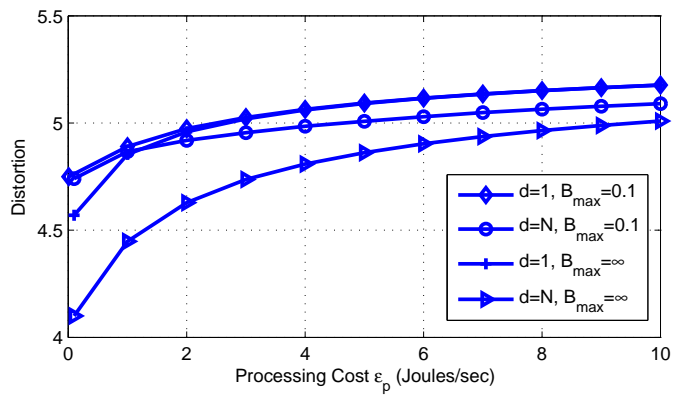}
\caption{Total distortion $D$ versus processing energy cost for a battery-run system. $E=4$ Joules, $B_{max}=0.1$ bits per sample, $\epsilon_s=0$, $\mathbf{h}=[0.4, 0.2, 0.2, 0.5,$ $0.4, 0.6, 0.9, 0.3, 0.4, 1]$, $\mathbf{\sigma^2}=[0.7, 0.6, 1, 0.5, 0.3, 0.6, 0.2, 0.3, 0.7, 0.5]$.} \label{f:HH}
\end{figure}

Finally, we consider the effect of the sampling cost on the minimum distortion for a battery-run system illustrated in Fig. \ref{f:GG}. We set $E=4$ Joules and $\epsilon_p=0$. As seen in the figure, when the sampling cost is low, the effect of the limited data buffer on the sum achievable distortion is more significant. However, when we increase the sampling cost, the performance of the system is mostly determined by the delay constraint. As it can be seen from Fig. \ref{f:HH} and Fig. \ref{f:GG}, the behavior of the distortion with respect to sampling cost is similar to that of the processing cost.
\begin{figure}[ht]
\center
\includegraphics[width=4in,trim= 0 0 0 0]{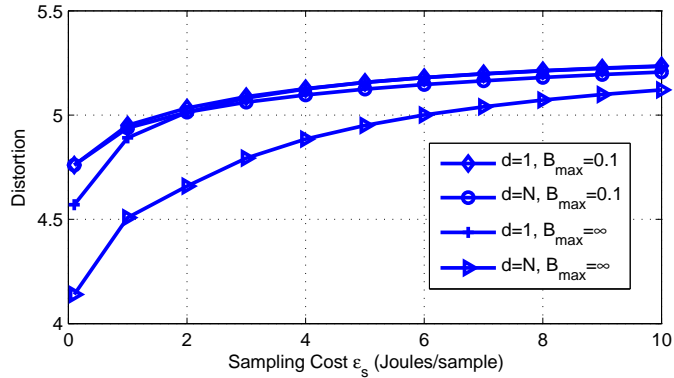}
\caption{Total distortion $D$ versus sampling energy cost for a battery-run system. $E=4$ Joules, $B_{max}=0.1$ bits per sample, $\epsilon_p=0$, $\mathbf{h}=[0.4, 0.2, 0.2, 0.5,$ $0.4, 0.6, 0.9, 0.3, 0.4, 1]$, $\mathbf{\sigma^2}=[0.7, 0.6, 1, 0.5, 0.3, 0.6, 0.2, 0.3, 0.7, 0.5]$.} \label{f:GG}
\end{figure}

\section{Conclusions} \label{s:conc}
We have investigated source-channel coding for a wireless sensor node under delay, data buffer size and various energy constraints. For a time slotted system, we have considered the scenario in which the samples of a time varying Gaussian source are to be delivered to a destination over a fading channel within $d$ time slots. In addition, we have imposed a finite size data buffer on the compressed samples. In this framework, we have investigated optimal transmission policies that minimize the total mean squared distortion of the samples at the destination for battery operated as well as an energy harvesting system. We have also studied the impact of various additional energy costs, including processing and sampling costs. In each case, we have provided a convex optimization formulation and identified the characteristics of the optimal distortion and power levels. We have also provided numerical results to investigate the impact of energy harvesting, processing and sampling costs. Our results have shown that for an energy harvesting transmitter energy arrivals over time may result in higher average distortion at the destination. In addition, we have observed that relaxing the delay and data buffer constraints induce more dramatic increase in the average distortion when processing and sampling costs are low. These results have important implications for the design of energy-limited wireless sensor nodes, and indicate that the optimal system operation and performance can be significantly different when the energy consumption of various other system components, or the arrival of the energy over time are taken into consideration.

\appendix{}\label{appendix}
In this appendix, we illustrate Fourier-Motzkin elimination of (\ref{rate 1})-(\ref{rate 3}) for three time slots $N=3$ when delay constraint is $d=2$. Rewriting (\ref{rate 1})-(\ref{rate 3}) in terms of $r_i \triangleq \frac{1}{2}\log\left(\frac{\sigma_i^2}{D_i}\right)$ and $c_i \triangleq \frac{1}{2}\log\left(1+h_i p_i\right)$ we get
\begin{eqnarray}
R_{1,1} &\leq & c_1 \nonumber \\
R_{1,2} + R_{2,2} &\leq & c_2 \nonumber \\
R_{2,3} + R_{3,3} &\leq &c_3 \nonumber \\
r_1 &\leq& R_{1,1} + R_{1,2} \nonumber \\
r_2 &\leq& R_{2,2} + R_{2,3} \nonumber \\
r_3 &\leq& R_{3,3} \nonumber \\
R_{1,1} +R_{1,2} &\leq & B_{max} \nonumber \\
R_{1,2} +R_{2,2} +R_{2,3} &\leq &B_{max} \nonumber \\
R_{2,3} +R_{3,3} &\leq & B_{max}, \nonumber
\end{eqnarray}
where $R_{1,1} \geq 0$, $R_{1,2} \geq 0$, $R_{2,2} \geq 0$, $R_{2,3} \geq 0$, $R_{3,3} \geq 0$, $r_i \geq 0$, and $c_i \geq 0$.

We have upper and lower bounds on $R_{1,1}$ as $\max\{0,r_1 -R_{1,2}\} \leq R_{1,1} \leq \min\{c_1, B_{max}-R_{1,2}\}$. Therefore, eliminating $R_{1,1}$ and the redundant inequalities, we obtain:
\begin{eqnarray}
r_1 &\leq & c_1 + R_{1,2} \nonumber \\
R_{1,2} + R_{2,2} &\leq & c_2 \nonumber \\
R_{2,3} + R_{3,3} &\leq &c_3 \nonumber \\
r_2 &\leq& R_{2,2} + R_{2,3} \nonumber \\
r_3 &\leq& R_{3,3} \nonumber \\
r_1 &\leq & B_{max} \nonumber \\
R_{1,2} +R_{2,2} +R_{2,3} &\leq &B_{max} \nonumber \\
R_{2,3} +R_{3,3} &\leq & B_{max} \nonumber
\end{eqnarray}

The upper and lower bounds on $R_{1,2}$ are $\max\{0,r_1 -c_1\} \leq R_{1,2} \leq \min\{c_2-R_{2,2}, B_{max}- R_{2,2}-R_{2,3}\}$. Therefore, eliminating $R_{1,2}$ and the redundant inequalities, we obtain:
\begin{eqnarray}
r_1 +  R_{2,2} &\leq & c_1 +c_2  \nonumber \\
R_{2,2} &\leq & c_2 \nonumber \\
R_{2,3} + R_{3,3} &\leq & c_3 \nonumber \\
r_2 &\leq& R_{2,2} + R_{2,3} \nonumber \\
r_3 &\leq& R_{3,3} \nonumber \\
r_1 +  R_{2,2} + R_{2,3}&\leq & c_1 +B_{max}  \nonumber \\
r_1 &\leq & B_{max} \nonumber \\
R_{2,2}+R_{2,3} &\leq & B_{max} \nonumber \\
R_{2,3} +R_{3,3} &\leq & B_{max} \nonumber
\end{eqnarray}

The upper and lower bounds on $R_{2,2}$ are $\max\{0,r_2 -R_{2,3}\} \leq R_{2,2} \leq \min\{c_2, c_1 + c_2-r_1, B_{max}- R_{2,3}, c_1+ B_{max}-r_1 - R_{2,3}\}$. Eliminating $R_{2,2}$ and the redundant inequalities,  we obtain:
\begin{eqnarray}
r_1 &\leq & c_1+  c_2  \nonumber \\
r_2  &\leq & c_2 + R_{2,3} \nonumber \\
r_1 +r_2 &\leq & c_1+  c_2  + R_{2,3} \nonumber \\
R_{2,3} + R_{3,3} &\leq & c_3 \nonumber \\
r_3 &\leq& R_{3,3} \nonumber \\
r_i &\leq & B_{max}, \quad i=1,2 \nonumber \\
r_1 + R_{2,3} &\leq & B_{max} +c_1 \nonumber \\
r_1 + r_2  &\leq & B_{max} +c_1 \nonumber \\
R_{2,3} +R_{3,3} &\leq & B_{max} \nonumber
\end{eqnarray}

The upper and lower bounds on $R_{2,3}$ are $\max\{0,r_2 -c_2, r_1+ r_2 -c_1 -c_2\} \leq R_{2,3} \leq \min\{B_{max} + c_1 -r_1, c_3 - R_{3,3}, B_{max}-R_{3,3}\}$. Eliminating $R_{2,3}$ and the redundant inequalities, we obtain:
\vspace{-0.4in}
\begin{eqnarray}
R_{3,3} + r_1+r_2 &\leq &  c_3+c_2 +c_1\nonumber \\
r_1 &\leq & c_1+  c_2  \nonumber \\
R_{3,3} &\leq &  c_3\nonumber \\
R_{3,3} + r_2 &\leq &  c_3+c_2\nonumber \\
r_3 &\leq& R_{3,3} \nonumber\\
r_i &\leq & B_{max}, \quad i=1,2 \nonumber \\
R_{3,3} &\leq &  B_{max}\nonumber \\
R_{3,3} + r_2 &\leq &  B_{max}+c_2\nonumber \\
R_{3,3} +r_1+ r_2 &\leq &  B_{max}+c_2+ c_1\nonumber \\
r_1 + r_2  &\leq & B_{max} +c_1 \nonumber
\end{eqnarray}
\vspace{-0.41in}

Finally, we have upper and lower bounds on $R_{3,3}$ as $\max\{0, r_3\} \leq R_{3,3} \leq \min\{c_3,c_3+c_2-r_2,B_{max}, B_{max}+c_2-r_2, B_{max}+c_1+c_2 - r_1 -r_2, c_3+c_2 +c_1 -r_1-r_2\}$. Eliminating $R_{3,3}$ and the redundant inequalities, we obtain:
\vspace{-0.3in}
\begin{eqnarray}
r_3 &\leq &  c_3\nonumber \\
r_2 +r_3  &\leq &  c_2+c_3 \nonumber \\
r_1+ r_2 +r_3 &\leq & c_1+c_2 +c_3 \nonumber \\
r_1 &\leq & c_1+  c_2  \nonumber \\
r_1 + r_2  &\leq & c_1 + B_{max} \nonumber \\
r_1 + r_2 + r_3  &\leq & c_1 + c_2 + B_{max} \nonumber \\
r_2+r_3  &\leq & c_2 + B_{max} \nonumber \\
r_i &\leq & B_{max}, \quad i=1,2,3. \nonumber
\end{eqnarray}

\end{document}